\documentclass{amsart}
\usepackage{amsthm}
\usepackage{natbib}
\usepackage{amssymb}
\usepackage{amsfonts}
\usepackage{mathtools}
\usepackage{bbold} 
\usepackage{url}

\newcommand{\setX}{\mathcal{X}} 
\newcommand{\setZ}{\mathcal{Z}} 
\newcommand{\R}{\mathbb{R}} 

\newcommand{\dd}{\mathrm{d}}
\newcommand{\dx}{\dd x}
\newcommand{\dz}{\dd z}
\newcommand{\eqdef}{:=}
\newcommand{\ind}{\mathbb{1}}  
\newcommand{\V}{\mathcal{V}}

\newcommand{\E}{\mathbb{E}}
\newcommand{\Var}{\mathrm{Var}}
\newcommand{\Cov}{\mathrm{Cov}}
\renewcommand{\P}{\mathbb{P}}
\newcommand{\Q}{\mathbb{Q}}

\newcommand{\vas}{v_{\infty}}

\newcommand{\N}{\mathcal{N}}

\newcommand{\bigO}{\mathcal{O}}

\newcommand{\cvd}{\Rightarrow}
\newcommand{\cvas}{\stackrel{a.s.}{\rightarrow}}

\usepackage[ruled, vlined]{algorithm2e}

\newtheorem{thm}{Theorem}
\newtheorem{lem}{Lemma}
\newtheorem{prop}{Proposition}
\newtheorem*{assumption}{Assumption}
\newtheorem{defi}{Definition}

\newenvironment{algo}[1]{
  \begin{center}
      \begin{algorithm}
        \caption{#1}
        \DontPrintSemicolon
      }
      {
      \end{algorithm}
  \end{center}
}

\newcommand{\FK}{Feynman-Kac}
\newcommand{\wf}{\mathrm{wf}}

\newcommand{\norm}[1]{\left\lVert#1\right\rVert}
\newcommand{\pr}[1]{\left( #1 \right)} 
\newcommand{\ps}[1]{\left[ #1 \right]} 
\newcommand{\px}[1]{\left\{ #1 \right\}} 
\newcommand{\abs}[1]{\left| #1 \right|} 
\newcommand{\CE}[2]{\E\ps{\left. {#1} \right \vert {#2}}} 
\newcommand{\CVar}[2]{\Var\pr{\left. {#1} \right \vert {#2}}} 

\newcommand{\tv}[1]{\norm{#1}_{\operatorname{TV}}}
\newcommand{\infnorm}[1]{\norm{#1}_\infty}

\newcommand{\wstd}[2]{\mathcal{V}^{\textrm{std,}#2}_{#1}}
\newcommand{\wwf}[1]{\mathcal{V}^{\textrm{wf}}_{#1}}
\newcommand{\ifactorstd}[2]{\operatorname{IF}_{#1}^{\textrm{std,}#2}}
\newcommand{\ifactorwf}[1]{\operatorname{IF}_{#1}^{\textrm{wf}}}

\author{Hai-Dang Dau \& Nicolas Chopin}

\title[Waste-free SMC]{Waste-free sequential Monte Carlo}

\begin{document}

\begin{abstract}
  A standard way to move particles in an SMC sampler is to apply several steps of
  an MCMC (Markov chain Monte Carlo) kernel. Unfortunately, it is not clear how
  many steps need to be performed for optimal performance. In addition, the
  output of the intermediate steps are discarded and thus wasted somehow. We
  propose a new, waste-free SMC algorithm which uses the outputs of all these
  intermediate MCMC steps as particles. We establish that its output is
  consistent and asymptotically normal. We use the expression of the asymptotic
  variance to develop various insights on how to implement the algorithm in
  practice. We develop in particular a method to estimate, from a single run of
  the algorithm, the asymptotic variance of any particle estimate. We show
  empirically, through a range of numerical examples, that waste-free SMC tends
  to outperform standard SMC samplers, and especially so in situations where the
  mixing of the considered MCMC kernels decreases across iterations (as in
  tempering or rare event problems).
\end{abstract}

\maketitle

\section{Introduction}
\label{sec:intro}

\subsection{Background}

Sequential Monte Carlo (SMC) methods are iterative stochastic algorithms that
approximate a sequence of probability distributions through successive
importance sampling, resampling, and Markov steps. Historically, they were
mainly used to approximate the filtering distributions of a state-space model.
More recently, they have been extended to an arbitrary sequence of probability
distributions \citep{MR1837132, MR1929161, DelDouJas:SMC}; in such applications,
they are often called ``SMC samplers''.

As an illustrative example, consider the tempering sequence:
\begin{equation}
  \label{eq:tempering_seq}
  \pi_t(\dx) \propto  \nu(\dx) L(x)^{\gamma_t}
\end{equation}
based on increasing exponents, $0=\gamma_0<\ldots <\gamma_T=1$. This sequence
may be used to interpolate between a distribution $\nu(\dx)$, which is easy to
sample from, and a distribution of interest, $\pi(\dx) \propto \nu(\dx) L(x)$
(e.g.\ a Bayesian posterior distribution), which may be difficult to simulate
directly. Other sequences of interest will be discussed later.

When used to sample from a fixed distribution (as in tempering), SMC samplers
present several advantages over MCMC (Markov chain Monte Carlo). First, they
provide an estimate of the normalising constant of the target distribution at no
extra cost; this quantity is of interest in several cases, in particular in
Bayesian model choice \citep[e.g.][]{MR3533634}. Second, they are easy to
parallelise, as the bulk of the computation treats the $N$ particles
independently \citep{Lee2010}. Third, it is easy to make SMC samplers
``adaptive''; that is, to use the current particle sample to automate the choice
of most of its tuning parameters. This is often crucial for good performance.

To elaborate on the third point, a common strategy to move the particles is to
apply a $k-$fold MCMC kernel that leaves the current distribution $\pi_t$
invariant. One may use for instance a random walk Metropolis kernel, with the
covariance of the proposal set to a small multiple of the empirical covariance
of the particle sample. In that way, the algorithm automatically scales to the
current distribution.

However, one tuning parameter of SMC samplers that is often overlooked in the
literature is the number $k$ of MCMC steps that should be applied to move the
particles. For instance, \citet{MR3634307} set $k=3$ arbitrarily in their
numerical experiments, but it turns out that this value is very sub-optimal, as
we show in our first numerical example.

A second issue with $k$ is that there is no reason to set it to a fixed value
across iterations. In application such as tempering, $\pi_t$ may become more and
more difficult to explore through MCMC; thus $k$ should be increased
accordingly, and may become very large.

To deal with these two issues, one could set $k$ adaptively; that is, iterate
MCMC steps until a certain stability criterion is met 
\citep{drovandipettitt2011, MR3283917, MR3515028, salomone2018unbiased,
Buchholz2018}.  However, in our experience, these approaches are not always
entirely reliable. There seems to be a fundamental difficulty in determining,
after $k$ steps have been performed, that this value of $k$ is optimal, without
performing several extra steps.

A third, and perhaps more essential issue, is that, if indeed large values of
$k$ are required for good performance, the intermediate output of these $k$ MCMC
steps are not used directly, and seems somehow wasted.

\subsection{Motivation and plan}

These issues motivated us to develop a waste-free SMC algorithm that exploits
the intermediate outputs of these MCMC steps; see Section~\ref{sec:proposed}.
The basic idea is to resample only $M=N/P$ out of the $N$ previous particles,
for some $P\geq 2$. Then each resampled particle is moved $P-1$ times through
the chosen MCMC kernel. The resampled particles and their $P-1$ iterates are
gathered to form a new sample of size $N$.

Standard results on the convergence of SMC estimates cannot be applied directly
to this new algorithm. We were able nonetheless to establish the consistency and
asymptotic normality of the output of waste-free SMC; see Section~\ref{sec:theory}. We also compared the performance and the robustness of waste-free SMC and standard SMC through an artificial example.

These theoretical results (in particular the expression of the asymptotic
variance) gives us various insights on how to implement waste-free SMC in
practice; see Section~\ref{sec:prac}. In particular, we are able to derive
variance estimates and confidence intervals for any particle estimate, which may
be computed from a single run.

To assess the performance and versatility of waste-free SMC, we perform
numerical experiments in three different scenarios where SMC samplers already
give state-of-the-art performance: logistic regression with a large number of
predictors; the enumeration of Latin squares; and the computation of Gaussian
orthant probabilities; see Section~\ref{sec:num}. In each case, waste-free SMC
performs at least as well as properly tuned SMC samplers, while requiring considerably less tuning effort.

Proofs are delegated to the appendix.

\subsection{Related work}

We focus on SMC samplers based on invariant (MCMC) kernels. These algorithms
have proved popular recently in a variety of applications, such as  
rare events \citep{johansen2005sequential, MR2909622}; 
experimental designs \citep{MR2281248};  cross-validation \citep{MR2676929};
variable selection  \citep{schafer2011sequential};  graphical models
\citep{NIPS2014_5570}; PAC-Bayesian classification  \citep{NIPS2014_5604}; 
Gaussian orthant probabilities \citep{MR3515028}; 
Bayesian model choice in hidden Markov models \citep{MR3533634}, and
un-normalised models \citep{MR3599680};
among others.

We note in passing that SMC samplers may be generalised to non-invariant
kernels, as shown in \cite{DelDouJas:SMC}; see also \cite{Heng2020} for how to
calibrate such kernels.  On the other hand,  it is also possible to add MCMC
steps to various SMC algorithms that are not SMC samplers; the idea goes back
to  \cite{MR1615251}. In particular, SMCMC   \citep[Sequential
MCMC,][]{septier2009mcmc, septier2016, MR4123640} algorithms approximate
recursively the filtering distribution of a state-space model: each iteration
$t$ runs a MCMC chain that leaves invariant a certain (partly discrete)
approximation of the current filter. It is not clear however how to derive a
waste-free version of these algorithms, and thus we do not consider them
further.

Finally, we mention that several improvements proposed for standard SMC
samplers might be also adapted to waste-free SMC, such as methods  to combine
the output of the intermediate steps,  see  \cite{Beskos2017} and
\cite{MR3960770}.

\cite{MR3357384} proposes several algorithms that are variations of the
resample-move algorithm of \cite{GilksBerzu}; one of them (generalized
resample-move) bears a similarity with waste-free SMC in the context of
of tempering. 
 
\section{Proposed algorithm}
\label{sec:proposed}

\subsection{Notations}
\label{subsec:notations}

Throughout the paper, $(\setX, \mathbb{X})$ stands for a measurable space, and
$\varphi:\setX\rightarrow \R$ for a measurable function; let
$\infnorm{\varphi}\eqdef \sup_{x\in\setX}|\varphi(x)|$ (supremum norm). The
expectation of $\varphi(X)$ when $X\sim\pi(\dx)$ is denoted by $\pi(\varphi)$;
i.e. $\pi(\varphi)\eqdef \int \varphi(x) \pi(\dx)$. Recall that a Markov kernel
$K(x, \dd y)$ is a map $K:\setX\times \mathbb{X} \rightarrow [0, 1]$ such that
$x\rightarrow K(x, A)$ is measurable in $x$, for any $A\in\mathbb{X}$; and
$A\rightarrow K(x, A)$ is a probability measure (on $(\setX, \mathbb{X})$), for
any $x\in\setX$. We use the following standard notations for the integral
operators associated to Markov kernel $K$: $\pi K$ is the distribution such that
$\pi K(A) = \int_\setX \pi(\dx) K(x, A)$, and $K(\varphi)$ is the
function $x\rightarrow \int_\setX K(x,\dd y) \varphi(y)$, for
$\varphi:\setX\rightarrow \R$.

Symbol $\cvd$ means convergence in distribution, and $\tv{\cdot }$ stands for
the total variation norm, $\tv{\mu-\nu}=\sup_{A\in\mathbb{X}}|\mu(A)-\nu(A)|$.
\subsection{A generic SMC sampler}
\label{sec:basic_smc_sampler}

We consider a generic sequence of target probability distributions of the form
(for $t=0, 1, \ldots, T$):
\begin{equation}
  \label{eq:generic_target}
  \pi_t(\dx) = \frac{1}{L_t} \gamma_t(x) \nu(\dx) 
\end{equation}
where $\nu(\dx)$ is a probability measure, with respect to measurable space
$(\setX, \mathbb{X})$, $\gamma_t$ is a measurable, non-negative function, and
$L_t\eqdef \int_{\setX}\gamma_t(x)\nu(\dx)$, the normalising constant, is
assumed to be properly defined, i.e. $0<L_t<\infty$. In the tempering scenario
mentioned in the introduction, $\gamma_t(x)=L(x)^{\gamma_t}$, for certain
exponents $\gamma_t$. Other interesting scenarios include data tempering
(sequential learning), where $x$ represents a parameter, $\nu(\dx)$ its prior
distribution, and $\gamma_t(x)$ is the likelihood of data-points $y_0,\ldots,
y_t$; rare-event simulation (and likelihood-free inference), where
$\gamma_t(x)=\ind_{\mathcal{E}_t}(x)$, the indicator function of nested sets
$\mathcal{E}_0\supset \mathcal{E}_1 \supset \ldots $; among others. See e.g.
Chapter 3 of \cite{SMCbook} for a review of common applications of SMC samplers,
and the sequence of target distributions arising in these applications.

One way to track the sequence $\pi_t$ would be to perform sequential importance
sampling: sample particles (random variates) from the initial distribution
$\nu(\dx)$, then reweight them sequentially according to weight function
$G_t(x)\eqdef\gamma_t(x) / \gamma_{t-1}(x)$ (for $t\geq 1$, and $G_0(x)\eqdef
\gamma_0(x)$). In most applications however, the weights degenerate quickly,
making this naive approach useless.

SMC samplers alternate such reweighting steps with resampling and Markov steps.
For the latter, we introduce Markov kernels $M_t(x_{t-1}, \dx_t)$ which leave
invariant the target distributions: $\pi_{t-1} M_t = \pi_{t-1}$ for $t \geq 1$.
It is easy to check that the sequence of \FK{} distributions (for $t=0,\ldots,
T$) defined as:
\begin{equation}
  \label{eq:FK}
  \Q_{t}(\dd x_{0:t}) = \frac{1}{L_t} \nu(\dx_0) \prod_{s=1}^t M_s(x_{s-1}, \dx_s)
  \prod_{s=0}^t G_s(x_s)
\end{equation}
is such that the marginal distribution of variable $X_t$ (with respect to
$\Q_t$) is $\pi_t$. We call \FK{} model the set of the components that define
this sequence of distributions, that is, the initial distribution $\nu$, the
kernels $M_t$, $t=1,\ldots,T$, and the functions $G_t$, $t=0,\ldots, T$. For
more background on \FK{} distributions, see e.g. \cite{DelMoral:book}.

Algorithm~\ref{alg:genericSMC} recalls the structure of an SMC sampler that
corresponds to this \FK{} model; and in particular which targets at each
iteration $t$ distribution $\pi_t$. It takes as inputs: $N$, the number of
particles, the considered \FK{} model, and the chosen resampling scheme
(function \texttt{resample}). Several resampling schemes exist. In this paper,
we focus for simplicity on multinomial resampling, which generates ancestor
variables $A_t^n$ independently from the categorical distribution that
generates label $m$ with probability $W_t^m$.

\begin{algo}{Generic SMC sampler\label{alg:genericSMC}}
  \KwIn{Integer $N\geq 1$, a \FK~model (initial distribution $\nu(\dx)$,
    functions $G_t$, Markov kernels $M_t$)} \For{$t\gets 0$ \KwTo $T$}{
    \If{$t=0$} { \For{$n=1$ \KwTo $N$} { $X_0^n\sim \nu(\dx_0)$ } } \Else{
      $A_t^{1:N} \sim \FuncSty{resample}(N, W_{t-1}^{1:N})$\; \For{$n=1$ \KwTo
        $N$} { $X_t^n \sim M_t(X_{t-1}^{A_t^n}, \dx_t)$ \; } } \For{$n \gets 1$
      \KwTo $N$} {$w_t^n \gets G_t(X_t^n)$ \; } \For{$n \gets 1$ \KwTo $N$} {
      $W_t^n \gets w_t^n / \sum_{m=1}^N w_t^m$ \; } }
\end{algo}

At any iteration $t$, quantity $\sum_{n=1}^N W_t^n \varphi(X_t^n)$ is an
estimate of the expectation $\pi_t(\varphi)$, for $\varphi:\setX\rightarrow \R$,
and quantity $L_t^N\eqdef \prod_{s=0}^t \ell_s^N$, where $\ell_s^N\eqdef
N^{-1}\sum_{n=1}^N w_s^n$, is an estimate of the normalising constant $L_t$.
These estimates are consistent and asymptotically normal (as $N\rightarrow
+\infty$) under general conditions.

\subsection{Note on the generality of Algorithm~\ref{alg:genericSMC}}
\label{sub:generality}

While generic, Algorithm~\ref{alg:genericSMC} is a simplified version of most
practical SMC samplers. In particular, we have stressed in the introduction the
importance of making SMC samplers adaptive; that is, to adapt both the
distributions $\pi_t$ and the Markov kernels $M_t$ on the fly. This means that
these quantities may depend on the current particle sample. For simplicity, our
notations do not account for this. We will see later that similar adaptation
tricks may be developed for waste-free SMC.

Another interesting generalisation is when the state space $\setX$ evolves over
time; in particular when its dimension increases. This happens for instance when
performing sequential inference on a model involving latent variables. The ideas
developed in this paper may easily be adapted to this scenario, as we shall see
in our third numerical example. For the sake of exposition, however, we focus on
the fixed state space case.


\subsection{Proposed algorithm: waste-free SMC}
\label{sec:proposed_algo}

The idea behind waste-free SMC is to resample only $M$ ancestors, with $M\ll N$.
Then each of these ancestors is moved $P-1$ times through Markov kernel $M_t$.
The resulting $M$ chains of length $P$ are then put together to form a new
particle sample, of size $N=MP$. See Algorithm~\ref{alg:wasteless}.

\begin{algo}{Waste-free SMC sampler\label{alg:wasteless}}
  \KwIn{Integers $M, P\geq 1$ (let $N\leftarrow MP$), a \FK~model (initial
    distribution $\nu(\dx)$, functions $G_t$, Markov kernels $M_t$)} \For{$t
    \gets 0$ \KwTo $T$}{ \If{$t=0$}{ \For{$n \gets 1$ \KwTo $N$} { $X_0^n \sim
        \nu(\dx_0)$ \; } } \Else{ $A_t^{1:M} \sim \FuncSty{resample}(M,
      W_{t-1}^{1:N})$\; \For{$m \gets 1$ \KwTo $M$}{ $\tilde{X}_t^{m,1} \gets
        X_{t-1}^{A_t^m}$\; \For{$p \gets 2$ \KwTo $P$}{ $\tilde{X}_t^{m,p} \gets
          M_t(\tilde{X}_t^{m,p-1}, \dx_t)$ \; } } Gather variables
      $\tilde{X}_t^{m,p}$ so as to form new sample $X_t^{1:N}$\; } \For{$n \gets
      1$ \KwTo $N$} { $w_t^n \gets G_t( X_t^n)$ \; } \For{$n \gets 1$ \KwTo $N$}
    { $W_t^n \gets w_t^n / \sum_{m=1}^N w_t^m$ \; } }
\end{algo}

The output of the algorithm may be used exactly in the same way as for standard
SMC: e.g. $\sum_{n=1}^N W_t^n \varphi(X_t^n)$ is an estimate of
$\pi_t(\varphi)$.

To get some intuition why waste-free SMC may be a valid and interesting
alternative to standard SMC, consider at time $t-1$ a fictitious particle
$X_{t-1}^n$, whose weight $W_{t-1}^n$ is large. In a standard SMC 
sampler, this particle is selected many times as an ancestor for the Markov
step. Then, if $M_t$ mixes poorly, its many children will be strongly
correlated.

On the other hand, in waste-less SMC, provided that $M\ll N$, the particle
$X_{t-1}^n$ is selected a much smaller number of times; each time it is
selected,  $P$ successive variables are introduced in the sample. By
construction, two such variables should be less correlated than if they had the
same ancestor (as in standard SMC); see Figure~\ref{fig:depend} for a
graphical representation of this idea.

\begin{figure}
  \centering
  \includegraphics[scale=0.45]{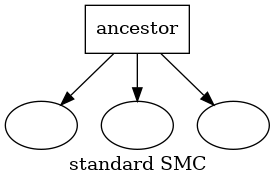}
  \includegraphics[scale=0.45]{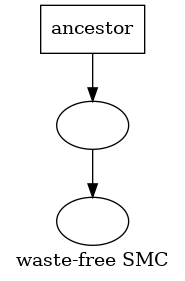}
  \caption{Pictorial representation of dependencies in standard SMC and
    waste-free SMC. Left: in standard SMC, an ancestor generates $3$ 
    children for the next iteration. Right: in waste-free SMC, the same ancestor
    generates itself, one child, and one grand-child. 
    Each arrow corresponds to one transition through kernel $M_t$. 
\label{fig:depend}
}
\end{figure}

Another insight is provided by chaos propagation theory \citep[][Chap.
8]{DelMoral:book}, which says that, when $M\ll N$, $M$ resampled particles
behave essentially like $M$ independent variables that follows the current
target distribution. Thus, in a certain asymptotic regime, we expect the
particle sample to behave like the variables of $M$ independent, stationary
Markov chains, of length $P$.

Before backing these intuitions with a proper analysis, we provide a last
insight regarding the underlying structure of waste-free SMC.

\subsection{\FK{} model associated with waste-free SMC}
\label{sec:fk_waste}

Algorithm~\ref{alg:wasteless} may be cast as a standard SMC sampler that
propagates and reweights particles that are Markov chains of length $P$. The
components of the corresponding \FK{} model may be defined as follows. Assume
$P\geq 1$ is fixed. Let $\setZ=\setX^P$, and, for $z\in\setZ$, denote component
$p$ as $z[p]$: $z=(z[1], \ldots, z[P])$. Then define the potential functions as:
\begin{equation}
  \label{eq:Gt_wf}
  G_t^{\wf}(z) \eqdef \frac{1}{P} \sum_{p=1}^P G_t(z[p])
\end{equation}
the initial distribution as: $\nu^\wf(\dz) \eqdef \prod_{p=1}^P \nu(\dz[p])$,
and the Markov kernels as:
\begin{multline}
  M_t^{\wf}(z_{t-1},\dz_t) \eqdef \\
  \left\{ \sum_{p=1}^P \frac{G_{t-1}(z_{t-1}[p])}{\sum_{q=1}^P
      G_{t-1}(z_{t-1}[q])} \times M_t(z_{t-1}[p], \dz_t[1])
  \right\}\prod_{p=2}^P M_t(z_t[p-1], \dz_t[p]).
\end{multline}

The following proposition explains how this waste-free \FK~model relates to the
initial \FK~model of Algorithm~\ref{alg:genericSMC}.

\begin{prop}
  \label{prop:wf_FK}
  The \FK~model associated with initial distribution $\nu^{\wf}$, Markov kernels
  $M_t^{\wf}$, and functions $G_t^\wf$, that is, the sequence of distributions:
  \[ \Q_t^{\wf}(dz_{0:t}) = \frac{1}{L_t^{\wf}} \nu^{\wf}(\dz_0) \prod_{s=1}^t
    M_s^{\wf}(z_{s-1}, \dz_s) \prod_{s=0}^t G_s^{\wf}(z_s)
  \]
  where $L_t^{\wf}$ is a normalising constant, is such that:
  \begin{itemize}
  \item $L_t^{\wf}=L_t$, the normalising constant of \eqref{eq:generic_target}
    and \eqref{eq:FK};
  \item $\Q_t^{\wf}(dz_t)$ is the distribution of a stationary Markov chain of
    size $P$ whose Markov kernel is $M_t$ (and thus whose initial distribution
    is $\pi_t$):
    \[ \Q_t^{\wf}(\dz_t) = \pi_t(\dz_t[1]) \prod_{p=2}^P
      M_t(z_t[p-1],\dz_t[p]). \]
  \end{itemize}
\end{prop}

We can now interpret Algorithm~\ref{alg:wasteless} as an instance of Algorithm
\ref{alg:genericSMC} where the number of particles is $M$, and the underlying
\FK~model is defined as above. In particular, consider how Algorithm
\ref{alg:genericSMC} would operate if applied to that \FK~model. At time $t$, it
would select randomly an ancestor $z_{t-1}$ (a chain of length $P$), with
probability $\propto \sum_{p=1}^P G_{t-1}(z_{t-1}[p])$. Then, when kernel
$M_t^{\wf}$ is applied to this chain, one component would be selected randomly,
with probability $G_{t-1}(z_{t-1}[p]) / \sum_{q=1}^P G_{t-1}(z_{t-1}[q])$. Thus,
this particular component would be used as a starting point of the subsequent
chain with probability $\propto G_{t-1}(z_{t-1}[p])$. This is precisely what is
done in Algorithm~\ref{alg:wasteless}.

This interpretation of waste-free SMC as a standard SMC sampler makes it easy
to derive several of its properties; for instance, regarding its estimates of
the normalising constants.

\begin{prop}
  \label{prop:unbiased}
  At iteration $t\geq 0$ of Algorithm~\ref{alg:wasteless} , the quantity
  \begin{equation}
    \label{eq:norm_const_est}
    L_t^N \eqdef 
    \prod_{s=0}^t \ell_s^N,\quad 
    \mbox{where } \ell_s^N \eqdef  \frac 1 N \sum_{n=1}^N G_s(X_s^n)
  \end{equation}
  is an unbiased estimate of $L_t$, the normalising constant of target
  distribution $\pi_t$, as defined in \eqref{eq:generic_target}.
\end{prop}

This proposition is a small variation over the well known result of
\citep{DelMoral1996unbiased} that, in a standard SMC sampler, the estimate of 
the normalising constant estimate is unbiased. 

We can also use the interpretation of waste-free SMC as a standard SMC sampler
to derive asymptotic results.

\begin{prop}
  \label{prop:clt_fixed_p}
  For $P\geq 1$ fixed, and $\varphi:\setX \rightarrow \R$ measurable and
  bounded, the output of Algorithm~\ref{alg:wasteless} at time $t\geq 0$ is such
  that
  \begin{align}
    \label{eq:clt_pred_fixedP}
    \sqrt{N} \left( \frac 1 N \sum_{n=1}^N \varphi(X_t^n) - \pi_{t-1}(\varphi)\right)
    & \cvd \N \left( 0, \tilde{\mathcal{V}}_t^P(\varphi) \right) \\
    \label{eq:clt_filt_fixedP}
    \sqrt{N} \left(  \sum_{n=1}^N W_t^n \varphi(X_t^n) - \pi_t(\varphi) \right)
    & \cvd \N \left( 0, \mathcal{V}_t^P(\varphi) \right) 
  \end{align}
  as $M\rightarrow +\infty$, $N=MP$, where $\pi_{t-1}$ means $\nu$ in
  \eqref{eq:clt_pred_fixedP} when $t=0$, $\tilde{\V}_0^P(\varphi)\eqdef
  \Var_{\nu}(\varphi)$,
  \begin{align}
    \tilde{\mathcal{V}}_t^P(\varphi)
    & \eqdef \V_{t-1}^P(\bar{M}_t^P\varphi) + v_P(M_t,\varphi),
      \quad t\geq 1, \label{eq:asympt_var_fixedP}
    \\
    \mathcal{V}_t^P(\varphi)
    & \eqdef \tilde{\mathcal{V}}_t^P\left(\bar{G_t} (\varphi-\pi_t\varphi) \right),
      \quad t\geq 0, 
  \end{align}
  $\bar{G}_t \eqdef G_t/\ell_t$, $\bar{M}_t^P=P^{-1}\sum_{p=1}^PM_t^{p-1}$,
  \begin{equation*}
    v_P(M_t,\varphi) \eqdef \Var\left( \frac 1 {\sqrt{P}} \sum_{p=0}^{P-1}\varphi(Y_p)
    \right)
  \end{equation*}
  and $(Y_p)_{p\geq 0}$ stands for a stationary Markov chain with kernel $M_t$
  (i.e. $Y_0\sim\pi_t$).
\end{prop}

This proposition is stated without proof, as it amounts to applying known
central limit theorems \citep[see Chapter 11 of][and references
therein]{SMCbook} for SMC estimates to the waste-free \FK{} model mentioned
above. Notice how the asymptotic variances depend on $P$ in a non-trivial way.
This suggests that the fixed $P$ regime is not very convenient; in particular it
is not clear how to choose $P$ for optimal performance. If we take
$P\rightarrow +\infty$, we expect the first term of \eqref{eq:asympt_var_fixedP}
to go to zero, and the second term to converge to the asymptotic variance of
kernel $M_t$. This suggests, at the very least, that taking $P$ large may often
be reasonable. The next section studies the asymptotic behaviour of the
algorithm as $P\rightarrow +\infty$.

\section{Convergence as $P\rightarrow +\infty$}
\label{sec:theory}

\subsection{Assumptions}

This section is concerned with the behaviour of waste-free SMC in the
``long-chain'' regime, that is, when $P\rightarrow +\infty$, while $M$ is either
fixed or may grow with $P$ at some rate. We start by remarking that this regime
requires some assumption on the mixing of the Markov kernels $M_t$. Indeed,
assume that $M_t$ is the identity kernel:
$M_t(x_{t-1},\dx_t)=\delta_{x_{t-1}}(\dx_t)$. In that case, at time $1$, one
has:
\begin{equation*}
  \frac 1 N \sum_{n=1}^N \varphi(X_1^n) = \frac 1 M \sum_{m=1}^M \varphi(X_0^{A_0^m})
\end{equation*}
since the $P$ particles $\tilde{X}_t^{m,p}$ are identical for a given $m$. The
variance of this quantity should be $\bigO(M^{-1})$, and cannot go to zero if
$M$ is kept fixed.
 
We thus consider the following assumptions.
\begin{assumption}[M]
  The Markov kernels $M_t$ are uniformly ergodic, that is, there exist constants
  $C_t \geq 0$ and $\rho_t\in[0,1[$ such that,
  \begin{equation*}
    \tv{M_t^k(x_{t-1}, \dx_t) - \pi_{t-1}(\dx_t) }
    \leq C_t \rho_t^k, \quad \forall x_{t-1}\in\setX, k\geq 1.
  \end{equation*}
\end{assumption}

\begin{assumption}[G]
  The functions $G_t$ are upper-bounded, $G_t(x)\leq D_t$ for some $D_t>0$ and
  all $x\in\setX$.
\end{assumption}

Ergodic Markov kernels in an SMC sampler was also considered in
\cite{Beskos2014} in order to study the behaviour of the algorithm as the
dimension of the state space gets high.

\subsection{Non-asymptotic bound}
We first state a non-asymptotic result.

\begin{prop}
  \label{thm:L2}
  Under Assumptions (M) and (G), there exist constants $c_t$ and $c_t'$ such
  that the following inequalities apply to the output of iteration $t \geq 0$ of
  Algorithm~\ref{alg:wasteless}, for any $M, P\geq 1$, and any bounded function
  $\varphi:\setX\rightarrow \R$:
  \begin{align}
    \E\left\{\frac 1 N \sum_{n=1}^N \varphi(X_t^n)
    - \pi_{t-1} (\varphi)\right\}^2
    & \leq c_t \frac{\infnorm{\varphi}^2}{N} \label{eq:L2_pred}\\ 
    \E\left\{\sum_{n=1}^N W_t^n \varphi(X_t^n)
    - \pi_t(\varphi)\right\}^2 
    & \leq c'_t \frac{\infnorm{\varphi}^2}{N} \label{eq:L2_filt}
  \end{align}
  where $\pi_{t-1}$ means $\nu$ in \eqref{eq:L2_pred} at time $t=0$.
\end{prop}

The constants $c_t$ and $c_t'$ are not sharp. However, this result remains
interesting, in that it shows that waste-free SMC is consistent (in $L^2$ norm,
and thus in probability) whenever $N=MP\rightarrow +\infty$, that is, whenever
$P\rightarrow +\infty$, or $M\rightarrow +\infty$, or both simultaneously,
possibly at different rates.

\subsection{Central limit theorems}
\label{sub:clts}
We now state a central limit theorem for the long chain regime.

\begin{thm}
  \label{thm:clt_longchain}
  Under Assumptions (M) and (G), for $M=M(P)=\bigO(P^\alpha)$, $\alpha\geq 0$
  (i.e. $M$ is either fixed or grows with $P$ at a certain rate) and
  $\varphi:\setX\rightarrow \R$ measurable and bounded, one has at any time
  $t\geq 0$
  \begin{align}
    \label{eq:clt_pred}
    \sqrt{N} \left( \frac 1 N \sum_{n=1}^N \varphi(X_t^n) - \pi_{t-1}(\varphi)\right)
    & \cvd \N \left( 0, \tilde{\mathcal{V}}_t(\varphi) \right) \\
    \label{eq:clt_filt}
    \sqrt{N} \left(  \sum_{n=1}^N W_t^n \varphi(X_t^n) - \pi_t(\varphi) \right)
    & \cvd \N \left( 0, \mathcal{V}_t(\varphi) \right) 
  \end{align}
  as $P\rightarrow \infty$ (or equivalently as $N\rightarrow \infty$, since
  $N=MP$), where $\pi_{t-1}$ in \eqref{eq:clt_pred} means $\nu$ at time $t=0$,
  $\tilde{\V}_0(\varphi)=\Var_\nu(\varphi)$,
  \begin{align}
    \label{eq:asymp_var_longchain}
    \tilde{\V}_t(\varphi)
    & \eqdef v_{\infty}(M_t, \varphi) \eqdef \Var\left( \varphi(Y_0) \right)
      + 2 \sum_{p=1}^\infty \Cov\left( \varphi(Y_0), \varphi(Y_p) \right),
      \quad t\geq 1,
    \\
    \V_t(\varphi)
    & \eqdef \tilde{\V}_t\left(\bar{G}_t (\varphi-\pi_t\varphi)\right), 
      \quad t\geq 0,
  \end{align}
  and $(Y_p)_{p\geq 0}$ stands for a stationary Markov chain with kernel $M_t$
  (hence $Y_0\sim \pi_t$).
\end{thm}

The most striking feature of the asymptotic variances above is that they depend
only on the current time step $t$; in standard CLTs for SMC algorithms, these
quantities are a sum of terms depending on all the previous time steps. More
precisely, $v_\infty(M_t,\varphi)$ is the asymptotic variance of an average
$P^{-1}\sum_{p=1}^P \varphi(Y_p)$ obtained from a single stationary Markov chain
with kernel $M_t$. Equation \eqref{eq:asymp_var_longchain} shows that the $N$
particles $X_t^n$ behave like $M$ independent, `long' Markov chains. This simple
interpretation will make it possible to construct estimates of the asymptotic
variances above; see Section~\ref{sub:variance_est}. We also note that these
asymptotic variances do not depend on $M$ (when $M$ is fixed), or its growth
rate (when $M=\bigO(P^\alpha)$, $\alpha>0$). This suggests that the performance
of the algorithm should depend weakly on the actual value of $M$, provided $M\ll
N$.

We now consider a similar result for the normalising constant estimates that may
be obtained from Algorithm~\ref{alg:wasteless}.
\begin{thm}
  \label{thm:clt_norm_cst}
  Under Assumptions (M) and (G), for $M=\bigO(P^\alpha)$, $\alpha\in[0,1)$ (i.e.
  either $M$ is fixed, or $M$ grows sub-linearly with $P$), and
  $\varphi:\setX\rightarrow \R$ measurable and bounded, one has at time $t\geq
  0$:
  \begin{equation}
    \label{eq:clt_norm_cst}
    \sqrt{N} \left( \log L_t^N - \log L_t \right) 
    \cvd \N \left( 0, \sum_{s=0}^t v_\infty(M_s,\bar{G}_s) \right)
  \end{equation}
  as $P\rightarrow \infty$ (or equivalently as $N\rightarrow \infty$ since
  $N=MP$).
\end{thm}

The theorem above puts a stronger constraint on $M$; i.e. it requires $M \ll P$,
and thus $M\ll N^{1/2}$ (while Theorem~\ref{thm:clt_longchain} requires only
$M\ll N$).

Note that
\[
  \log L_t^N - \log L_t = \sum_{s=0}^t \left(\log \ell_s^N - \log
    \ell_s\right),\quad \mbox{where } \ell_s^N = \frac 1 N \sum_{n=1}^N
  G_s(X_s^n),
\]
and we could already deduce from \eqref{eq:clt_pred} and the delta-method that
\[
  \sqrt{N} \left (\log \ell_s^N - \log \ell_s \right) \cvd \N\left( 0,
    v_\infty(M_s, \bar{G}_s) \right).
\]

Thus, \eqref{eq:clt_norm_cst} suggests that the error terms in this
decomposition are nearly independent. Again, we shall use this interpretation
to derive an estimate of the asymptotic variance of $L_t^N$.

\subsection{Comparing the asymptotic variances of standard and waste-free SMC}%
\label{sub:formal_comp}

In this sub-section, we use the previous results to  compare formally the
performance of standard SMC and waste-free SMC in an artificial example. 

Let $A_t$, $t=0, 1, \ldots$ be a sequence of subsets of $\setX$ such that $A_0
\supset A_1 \supset \ldots$ and $\nu(A_t) = r^t$ for some $r < 1$, and some
initial distribution $\nu$. Consider the \FK~distributions such that $G_t(x_t)
= \ind_{A_t}(x_t)$ and $M_t=K_t^k$, i.e. the $k-$fold kernel such that  $K_t(x,
B) = (1-p) \ind_B(x) + p  \pi_{t-1}(B)$ for some $0<p<1$.  (In words, with
probability $p$, do not move, with probability $1-p$, sample exactly from the
current target.)

A standard SMC sampler applied to this problem will fulfil a CLT of the form: 
\[\sqrt N \left(\sum_{n=1}^N W_t^n \varphi(X_t^n) - \pi_t(\varphi)\right) 
    \Rightarrow \mathcal{N}\left(0, \wstd{t}{k}(\varphi)\right); 
\]
see \eqref{eq:vtphi_smcbook} in the proof of Proposition \ref{prop:comparison}
for an expression for $\wstd{t}{k}(\varphi)$ and e.g. Chapter 11 of
\cite{SMCbook} for more details.  Define the `inflation factor' (relative
error) for standard SMC to be: 
\[\ifactorstd{t}{k}(\varphi) \eqdef
        \frac{\wstd{t}{k}(\varphi)}{\Var_{\pi_t}(\varphi)}.
    \] 

For waste-free SMC, we take $k=1$, i.e. $M_t=K_t$, and define similarly its
inflation factor to be 
$\ifactorwf{t}(\varphi) \eqdef \frac{\wwf{t}(\varphi)}{\Var_{\pi_t}(\varphi)}$,
where $\wwf{t}(\varphi)$ is the asymptotic variance defined in Theorem
\ref{thm:clt_longchain}.

\begin{prop} 
    \label{prop:comparison}
    For the model considered above, let $k_0:= \log r / 2\log(1-p)$, then 
\begin{enumerate} 
\item The quantities $\ifactorstd{t}{k}(\varphi)$ and $\ifactorwf{t}(\varphi)$ do not depend
    on $\varphi$.  

\item For the standard SMC sampler, the inflation factor
    $\ifactorstd{t}{k}$ is stable with respect to $t$ if and only if $k\geq k_0$. 
        If $k<k_0$ however, $\ifactorstd{t}{k}$ explodes exponentially with $t$.  

\item For the waste-free SMC sampler, $\ifactorwf{t}$ is  stable with respect 
        to $t$ and is always equal to $\frac{1}{r} \pr{\frac 2 p - 1}$.  

\item For any choice of $k$, we have
    \[\lim_{t\to \infty} \frac{\ifactorwf{t}}{k \ifactorstd{t}{k}} \leq 4. \]

\end{enumerate} 
\end{prop} 

In words, the performance of standard SMC may deteriorate very quickly whenever
the number of MCMC steps, $k$, is set to a too small value. On the other hand,
up to small factor, waste-free SMC provides the same level of performance as
standard SMC based on a well chosen value for $k$. 

Of course, these statements are proven here for a specific example;
however, our numerical experiments (Section~\ref{sec:num}) suggest they apply
more generally.

\section{Practical considerations}
\label{sec:prac}

\subsection{Choice of $M$}
\label{sub:choice_M}

By default, we recommend to take $M \ll N$, first, because our previous results
indicate that, within this regime, performance should be robust to the precise
value of $M$; and, second, because we observe empirically that this regime
usually leads to best performance (i.e. lowest variance for a given CPU budget).
See our numerical experiments in Section~\ref{sec:num}.

On parallel hardware, we recommend to take $M$ equal to, or larger than the number
of processors, as it is easy to divide the computational load of each iteration
of Algorithm~\ref{alg:wasteless} into $M$ independent tasks.

\subsection{Choice of kernels $M_t$}
\label{sub:thin}

As discussed in the introduction, a standard practice is to set $M_t$ to be a
$k-$fold Metropolis kernel, whose proposal is calibrated on the current particle
sample; e.g. for a random walk proposal, set the covariance matrix of the
proposal to a certain fraction of the empirical covariance matrix of the
particles.

This type of recipe may be used within waste-free SMC, with one important twist.
Contrary to standard SMC, we recommend to always take $k=1$. This recommendation
is based on the following thinning argument. We know from MCMC theory that
thinning (subsampling) an MCMC chain is generally detrimental: \citet[][Theorem
3.3]{geyer1992practical} shows that $k v_\infty(M_t^k,\varphi) > v_\infty(M_t,
\varphi)$ (provided $M_t$ is reversible and irreducible). In words, between two
estimates computed from the same long chain, one using all the samples, and the
other using only one every other $k$-sample, the former will have a lower
variance (asymptotically, as the length of the chain goes to infinity).

The same remark applies to waste-free SMC: if we compare a waste-free SMC
sampler with $N$ particles, and Markov kernels $M_t=K_t^k$, for a certain $K_t$,
with the same algorithm with $kN$ particles, and kernels $M_t=K_t$, then the
latter will have (asymptotically) lower variance, given the expression of the
asymptotic variances in Theorem~\ref{thm:clt_longchain}.

As announced in the introduction, we see therefore that waste-free SMC is
indeed more economical than standard SMC, as it is able to exploit all the
intermediate steps of a given MCMC kernel (while standard SMC often requires to
take $k\gg 1$ for optimal performance).

\subsection{Variance estimation from a single run}
\label{sub:variance_est}

As explained below Theorem~\ref{thm:clt_longchain}, the output of waste-free SMC
at time $t$ behaves asymptotically like $M$ independent, stationary chains of
size $P$. Thus, to estimate the asymptotic variance
$\tilde{V}_t(\varphi)=v_\infty(M_t,\varphi)$ in \eqref{eq:clt_pred}, we propose
the following `$M$-chain estimate'. Denote by $\gamma_{t,q}^{M,P}$ the empirical
autocovariance of order $q \in \px{0, 1, \ldots, p-1}$ computed from the $M$
chains:
\begin{equation*}
  \gamma_{t,q}^{M,P} \eqdef \frac{1}{MP} \sum_{m=1}^M\sum_{p=1}^{P-q}
  \ps{\varphi(\tilde X_t^{m,p}) -  \mu^{M,P}_t(\varphi)} 
    \ps{\varphi(\tilde X_t^{m, p+q}) - \mu^{M,P}_t(\varphi)}
\end{equation*}
where $\mu^{M,P}_t(\varphi) \eqdef N^{-1} \sum_{m=1}^M \sum_{p=1}^P
\varphi(\tilde X_t^{m,p})$ is the empirical mean. Then, the estimator is defined
as
\begin{equation*}
  \tilde{V}_t^{M, P}(\varphi) \eqdef 
   \psi_P\left(\gamma_{t,0}^{M,P}(\varphi) ,
    \ldots, \gamma_{t,P-1}^{M,P}(\varphi) \right)
\end{equation*}
where $\psi_P:\R^P\rightarrow \R$ is a certain estimator of the asymptotic
variance $v_\infty(M_t,\varphi)$ based on the autocorrelations of a single chain
of length $P$.

Several such single-chain estimators $\psi_P$ have been proposed in the
literature, see e.g. the introduction of \cite{MR2604704}. In our experiments,
we found the initial monotone sequence estimator of \cite{geyer1992practical} to
be a convenient default, as it is simple to use (no tuning parameter), and it
seems to work well. Note however that this estimator is based on a property
which is specific to reversible kernels (namely that sums of adjacent pairs of
autocovariance form a decreasing sequence). When the chosen kernels $M_t$ are
not reversible, one may consider an alternative estimator; see our third
numerical experiment (Section~\ref{sec:num}) for more discussion on this point.

To estimate $V_t(\varphi)=\tilde{V}_t(G_t(\varphi-\Q_t(\varphi)))$, we use the
same approach with $\varphi$ replaced by $G_t(\varphi-\Q_t^N(\varphi))$,
$\Q_t^N(\varphi)=\sum_{n=1}^N W_t^n \varphi(X_t^n)$. Similarly, to estimate each
term in the asymptotic variance of the log normalising constant,
\eqref{eq:clt_norm_cst}, we replace $\bar{G_t}=G_t/\ell_t$ by $G_t/\ell_t^N$.

We note that there is an alternative approach to obtain variance estimates from
a single run of waste-free SMC. It consists in (a) casting waste-free SMC as a
standard SMC sampler, as we did in Section~\ref{sec:fk_waste} (taking $P$
fixed); and (b) to apply the method of \cite{Lee2018}, see also \cite{Chan2013},
\cite{Olsson2019} and \cite{2019arXiv190913602D}, for obtaining variance
estimates from SMC outputs. This method relies on genealogy tracking (i.e.
tracking the ancestors at time 0 of each current particle).

This alternative approach has two drawbacks however. First, it relies on the
fixed $P$ regime, while, as already said, we recommend by default to run
waste-free SMC in the $P\rightarrow +\infty$ regime, i.e. by taking $M\ll N$.
Second, the method of \cite{Lee2018} degenerates as soon as the number of common
ancestors of the $N$ particle drops to one; something which tends to occur
quickly as $t$ increases.

One may mitigate the degeneracy by tracking the genealogy only up to time $t-l$,
for a certain lag value $l$, as recommended by \cite{Olsson2019}. However this
introduces a bias, and choosing $l$ is non-trivial.

We will compare both approaches in the numerical experiments of
Section~\ref{sec:num}.

\subsection{On-line adaptation of $P$}
\label{sub:adaptive}

In certain applications, the mixing of kernels $M_t$ may vary wildly with $t$;
for instance, for a tempering sequence, the mixing of $M_t$ may deteriorate over
time. The second numerical example in Section~\ref{sec:num} illustrates this
phenomenon.

In such a case, it makes sense to adjust the computational effort to the mixing
of the chain. That is, at time $t$, take $P=P_t$ so that the variance of
estimates computed at time $t$ stay of the same order of magnitude. In practice,
we found the following strategy to work reasonably well: at iteration $t$,
adjust $P_t$ so that it exceeds $\kappa$ times the auto-correlation time of
kernel $M_t$, i.e. the quantity $v_\infty(M_t, \varphi)/2\Var_{\pi_t}(\varphi)$
for a certain constant $\kappa\geq 1$, and a certain function $\varphi$, as
estimated from the current sample (which consists of $M$ chains of length
$P_t$). In our simulations, we took $\varphi=\log G_t$, and $\kappa$ between 2 and
10. To adjust $P_t$, we set it to an initial value, then we doubled it until the
requirement was met.

The main drawback of this adaptive approach is that it makes the CPU time of the
algorithm random, which is less convenient for the user. On the other hand, it
seems to present two advantages, as observed in our experiments (see second
example in Section~\ref{sec:num}): (a) it avoids the poor performance one
obtains by taking a value for $P$ that is too small for certain iterations $t$;
and (b) it makes the variance estimates more robust in this type of scenario.

\section{Numerical experiments}
\label{sec:num}

In this section, we evaluate the performance of waste-free SMC in a variety of
challenging scenarios, covering different types of state-spaces (continuous or
discrete, with a fixed or an increasing dimension), of sequence of target
distributions (based on tempering or something else), and of MCMC kernels
(Metropolis or Gibbs). In each example, standard SMC is known to be a
competitive approach, and we assess in particular how waste-free SMC may improve
on the performance of standard SMC.

\subsection{Logistic regression}

We consider the problem of sampling from, and computing the normalising constant
of, the posterior distribution of a logistic regression model, based on data
$(y_i,z_i)\in\{-1,1\}\times \R^p$, 
parameter $x\in \R^p$, and likelihood
\[
  L(x) = \prod_{i=1}^{n_\mathcal{D}} F(y_i x^T z_i), \quad F(x) =
  \frac{1}{1+e^{-x}}.
\]
We consider the sonar dataset (available in the UCI machine learning
repository), which is one of the more challenging datasets considered in
\cite{MR3634307}, and for which SMC tempering is one of the competitive
alternatives (and the only one that may be used to estimate the marginal
likelihood). Following standard practice, each predictor is rescaled to have
mean 0 and standard deviation $0.5$; an intercept is added; the dimension of
$\setX$ is then $p=63$. The prior is an independent product of centred normal
distributions, with standard deviation $20$ for the intercept, $5$ for
other coordinates.

We compare the performance of standard SMC and waste-free SMC when applied to
the tempering sequence $\pi_t(\dx) \propto \nu(\dx) L(x)^{\gamma_t}$. In both
cases, the tempering exponents are set automatically (using Brent's method) so
that the ESS of each importance sampling step equals $\alpha N$, and the Markov
kernel $M_t$ is a $k$-fold random walk Metropolis kernel calibrated to the
resampled particles (see Section~\ref{sub:thin}). For waste-free, we always
take $k=1$ (as per the thinning argument of the same Section).
We take $\alpha=1/2$ here; see the supplement for results with other values of 
$\alpha$.

Figure~\ref{fig:bin_sonar_logLT} plots box-plots of estimates of the log of the
normalising constant of the posterior obtained from 100 independent runs of
standard SMC, for $k=5$, $20$, $100$, $500$, and $1000$ and waste-free SMC for
$k=1$, and $M=50$, $100$, $200$, $400$ and $800$. The number of particles
is set to $N=N_0/k$, with $N_0=2 \times 10^5$, so that all algorithms have roughly the
same CPU cost. (For waste-free, $P$ is adjusted accordingly, i.e. $P=N/M$, with
$N=2 \times 10^5$.) Figure~\ref{fig:bin_sonar_QT} does the same for the estimate of the
posterior expectation of the mean of all components of $x$, namely $\pi_T(\varphi)$ with $\varphi(x) \eqdef p^{-1} \sum_{s=1}^p x_s$ for $x \in \mathbb{R}^p$.

\begin{figure}
  \centering \includegraphics[scale=0.55]{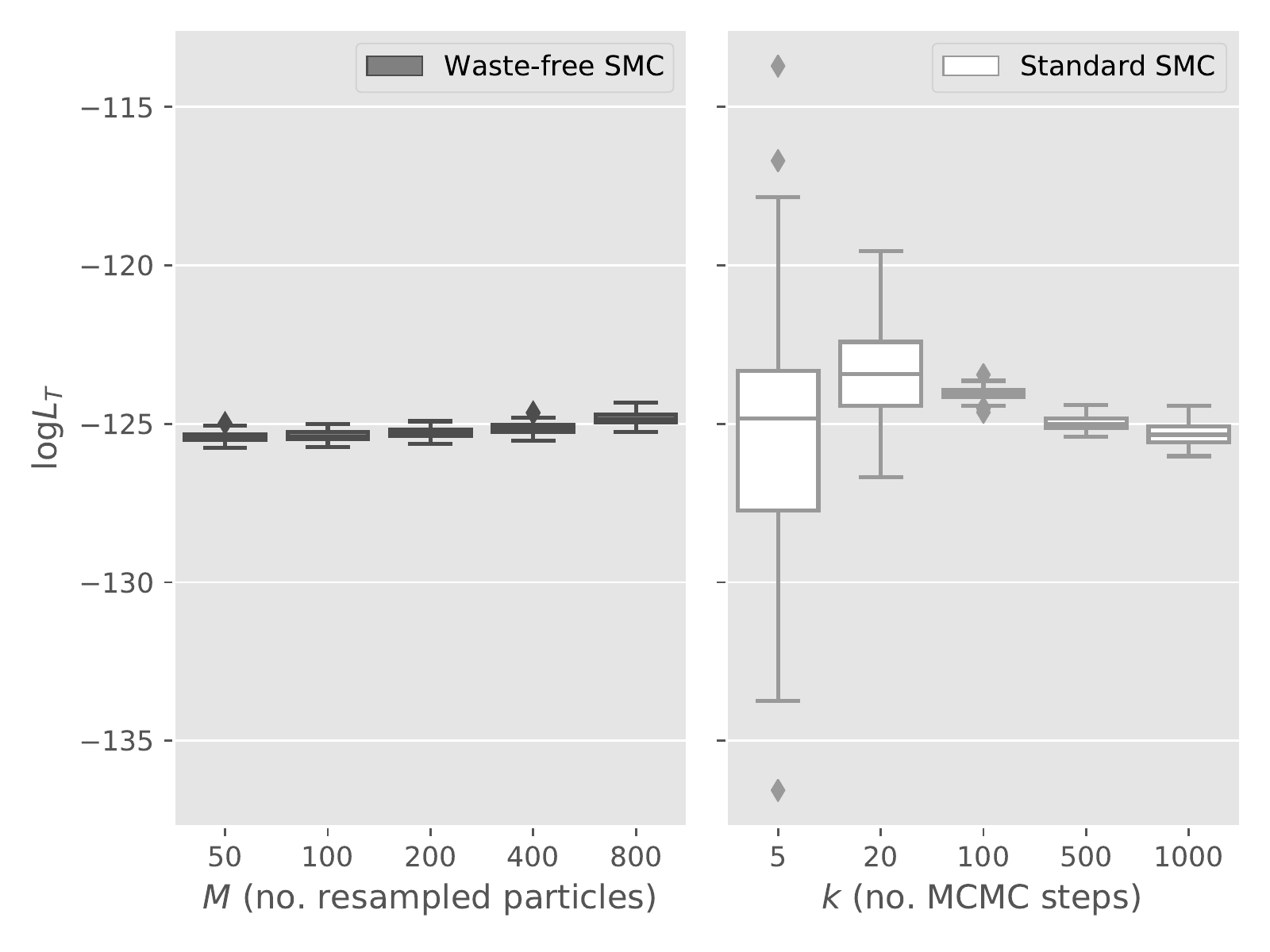}
  \caption{Logistic regression: estimates of the normalising constant obtained
    from waste-free SMC ($N=2 \times 10^5$) and standard SMC ($N=2 \times 10^5/k$). 
  \label{fig:bin_sonar_logLT}
}
\end{figure}

\begin{figure}
  \centering \includegraphics[scale=0.55]{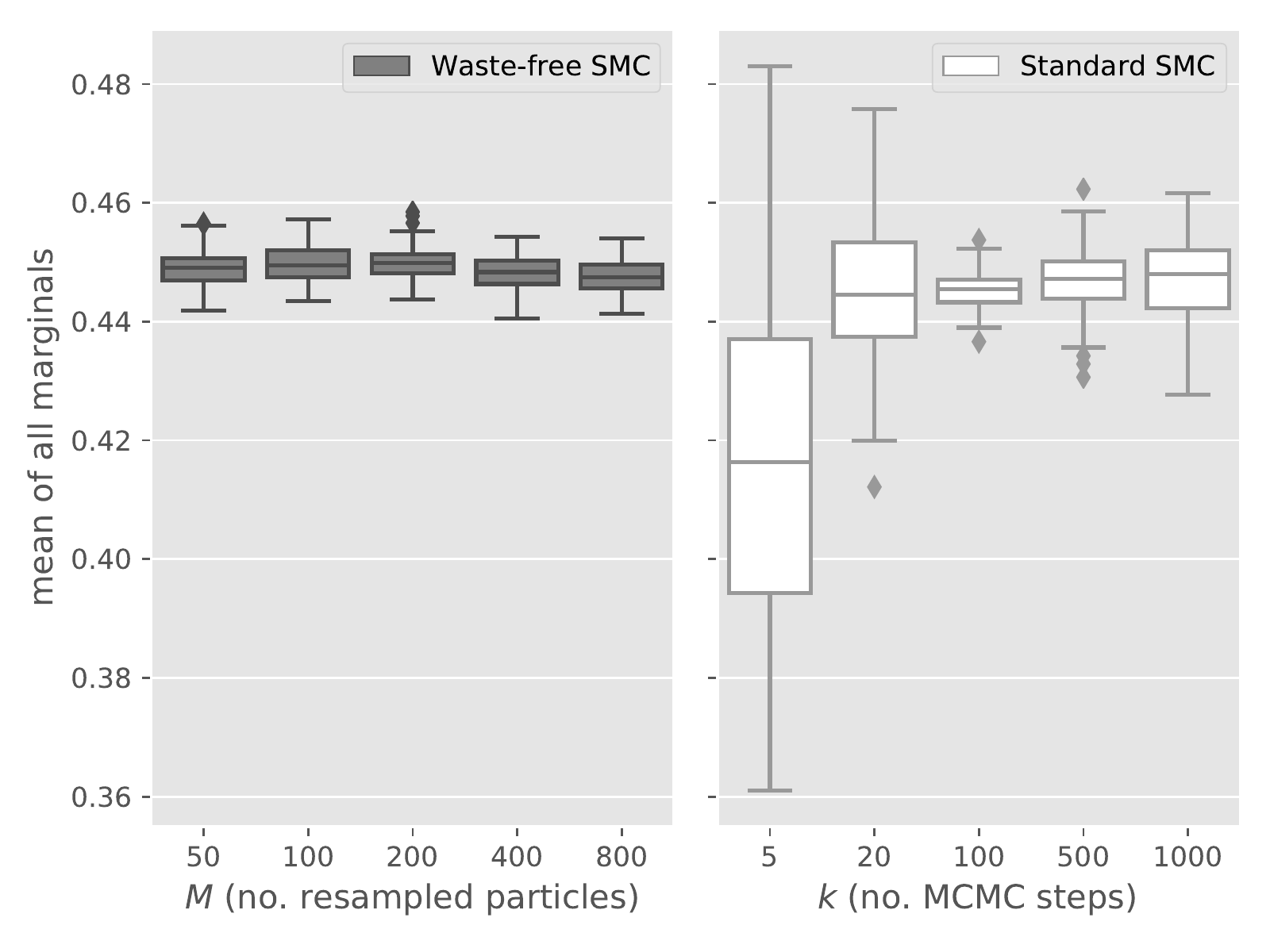}
  \caption{Same plot as Figure~\ref{fig:bin_sonar_logLT} for the estimate
    of the posterior expectation of the mean of all coordinates. 
    \label{fig:bin_sonar_QT}
  }
\end{figure}

These figures deserve several comments. First, waste-free seems to perform best
in the ``long chain'' regime, when $M\ll N$. Second, within this regime, the
performance seems robust to the choice of $M$; notice how the same level of
performance is obtained whether $M=50$ or $M=400$ (similar performance is also
obtained for $M < 50$, results not shown. We focused on $M \geq 50$ for reasons
related to parallel hardware as discussed in Section~\ref{sub:choice_M}.) Third,
in contrast, it seems difficult to choose $k$ to obtain optimal performance;
notice in particular that Figure~\ref{fig:bin_sonar_QT} suggests to take
$k=100$, but, for this value of $k$, the estimate of the log-normalising
constant seems biased, see Figure~\ref{fig:bin_sonar_logLT}. (Interestingly, we
observed such an upward bias for all values of $k$ when we ran standard SMC for
a smaller value of $N$, $N=10^5$; hence standard SMC seems also slightly less
robust to the choice of $N$; results not shown.) Fourth, and perhaps most
importantly, we are able to obtain better performance from waste-free SMC for a
given CPU budget.

We now evaluate the performance of the variance estimates discussed in
Section~\ref{sub:variance_est}. Figure~\ref{fig:bin_sonar_var} shows box-plots
of these estimates obtained from 100 runs of waste-free SMC, for $N=2 \times
10^5$ and $M=50$: the $M-$chain estimate advocated in
Section~\ref{sub:variance_est}; the estimate of \cite{Olsson2019}, with a lag
of 3 (the biased, but more stable version of \cite{Lee2018}, as explained in
Section~\ref{sub:variance_est}) and finally, the empirical variance over 10
independent runs. All these variance estimates are re-scaled by the same
factor, such that the empirical variance over the 100 runs equals one. (Other
values for the lag in the method of \cite{Olsson2019} did not seem to give
better results.)

Clearly, the $M-$chain estimator is more satisfactory, as it performs better
(especially for the normalising constant, left plot) than the empirical
variance, although being computed from a single run. On the other hand, the
approach of \cite{Lee2018} performs poorly. To be fair, this approach works more
reasonably if we increase significantly $M$ (results not shown), but since
taking $M$ too large decreases the performance of the algorithm, it seems fair to
state that this approach is not useful for waste-free SMC, at least in this
example.

\begin{figure}
  \centering \includegraphics[scale=0.35]{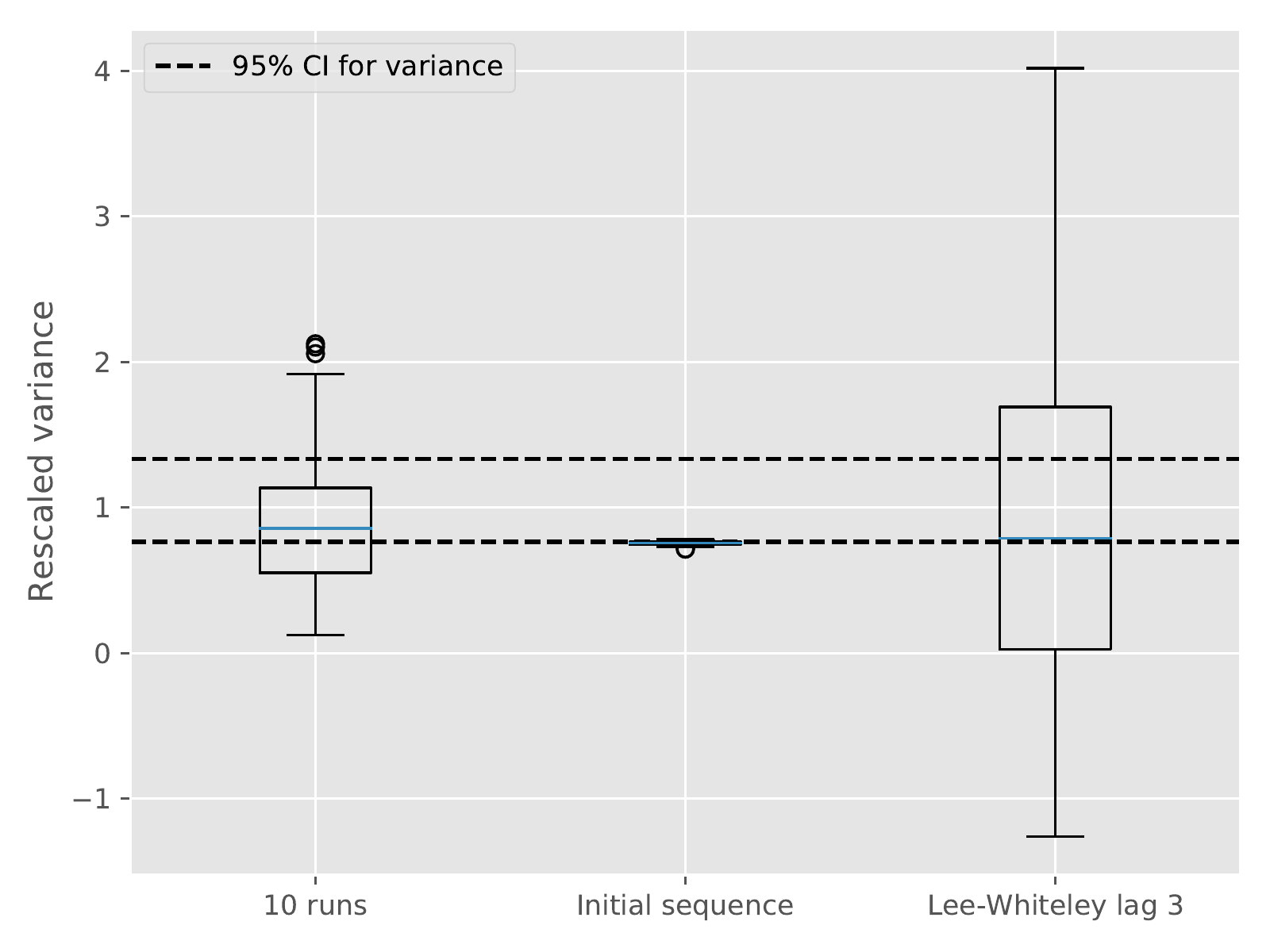}
  \includegraphics[scale=0.35]{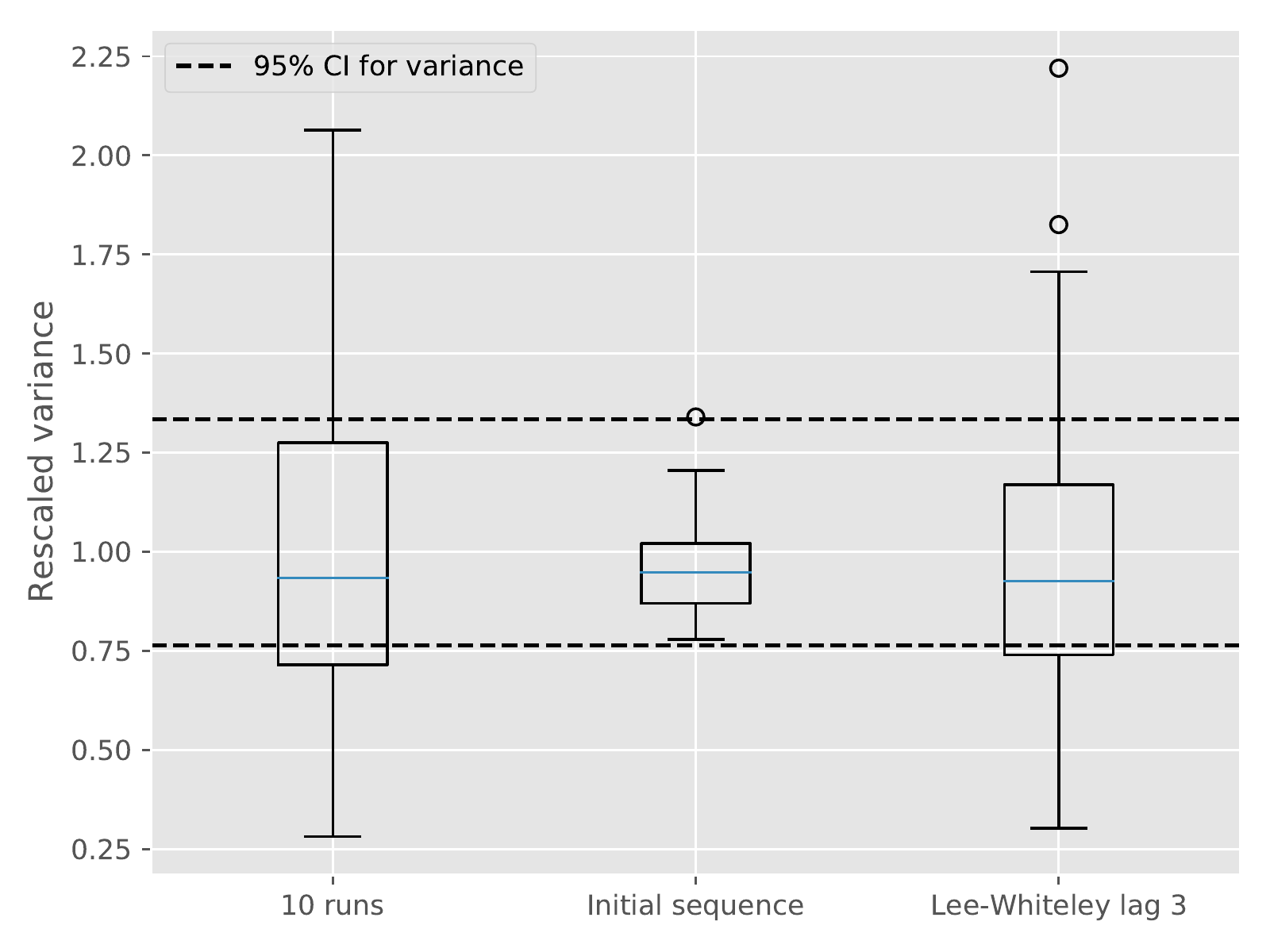}
  \caption{Logistic regression: box-plots of variance estimates over 100 runs
    obtained with waste-free SMC. Left: variance of the log-normalising constant
    estimate. Right: variance of the mean of all coefficients estimate. The
    variance estimates are re-scaled so that the empirical variance over the 100
    runs equals one; see text for more details. 
    \label{fig:bin_sonar_var}
  }
\end{figure}

\subsection{Latin squares}
\label{subset:lat}

Our second example concerns the enumeration of Latin squares of size $d$; that
is, $d\times d$ matrices with entries in $\{0, \ldots, d-1\}$, and such that
each integer in that range appears exactly once in each row and in each column;
see Table~\ref{tab:latin} for an example. The number $l(d)$ of Latin squares of
size $d$ increases very quickly with $d$, and is larger than $10^{43}$ for
$d=11$, the largest value for which it is known; see sequence A002860 of the
OEIS database \citep{Latin_squares_sequence}.

\begin{table}
  \label{tab:latin}
  \begin{tabular}{rrrrrrrrrr}
    1 &  5 &  0 &  3 &  7 &  8 &  9 &  6 &  2 &  4 \\
    0 &  4 &  5 &  8 &  6 &  9 &  1 &  7 &  3 &  2 \\
    2 &  8 &  7 &  0 &  9 &  4 &  5 &  3 &  1 &  6 \\
    3 &  7 &  4 &  1 &  5 &  2 &  8 &  0 &  6 &  9 \\
    6 &  0 &  9 &  5 &  1 &  3 &  2 &  8 &  4 &  7 \\
    8 &  2 &  1 &  9 &  4 &  0 &  6 &  5 &  7 &  3 \\
    9 &  6 &  3 &  2 &  0 &  5 &  7 &  4 &  8 &  1 \\
    5 &  1 &  6 &  4 &  3 &  7 &  0 &  2 &  9 &  8 \\
    4 &  9 &  2 &  7 &  8 &  6 &  3 &  1 &  5 &  0 \\
    7 &  3 &  8 &  6 &  2 &  1 &  4 &  9 &  0 &  5 \\
  \end{tabular}
  \caption{A Latin square of size $10$}
\end{table}

Let $\setX$ be the set of permutation squares of size $d$, that is, $d\times d$
matrices such that each row is a permutation of $\{0, \ldots, d-1\}$, and let
$p(d)$ its cardinal, $p(d)=(d!)^d$. We consider the following sequence of
tempered distributions: $\pi_t(\dx) = \nu(\dx) \exp\{ - \lambda_t V(x) \} / L_t
$, where $\nu(\dx)$ stands for the uniform distribution over $\setX$, and $V$ is
a certain score function such that $V(x)=0$ if $x$ is a Latin square, $V(x)\geq
1$ otherwise. Specifically, denoting the entries of matrix $x$ by $x[i, j]$, we
take
\[
  V(x) = \sum_{j=1}^d \left\{ \sum_{l=1}^d \left( \sum_{i=1}^d \ind(x[i,j]=l)
    \right)^2 - d \right\}.
\]

The quantity $L_t \times p(d)$ will be at distance $\varepsilon$ of $l(d)$, the
number of Latin squares, as soon as $\lambda_t \geq \log( p(d) /\varepsilon ) $.
Thus, we select adaptively the successive exponents $\lambda_t$ (as in the
previous example), and stop the algorithm at the first iteration $t$ such that
this condition is fulfilled, for $\varepsilon=10^{-16}$.

We set the Markov kernel $M_t$ to be a $k$-fold Metropolis kernel based on the
following proposal distribution: given $x$, select randomly a row $i$, two
columns $j$, $j'$, and swap components $x[i, j]$ and $x[i, j']$.

Figure~\ref{fig:latin_logLT} compares the performance of standard SMC and
waste-free SMC for evaluating the log of the normalising constant $L_T$, that is
(up to a small error as explained above), the log of the number of Latin squares
$l(d)$; we take $d=11$ since this is the largest value of $d$ for which $l(d)$
is known exactly.

\begin{figure}
  \centering \includegraphics[scale=0.55]{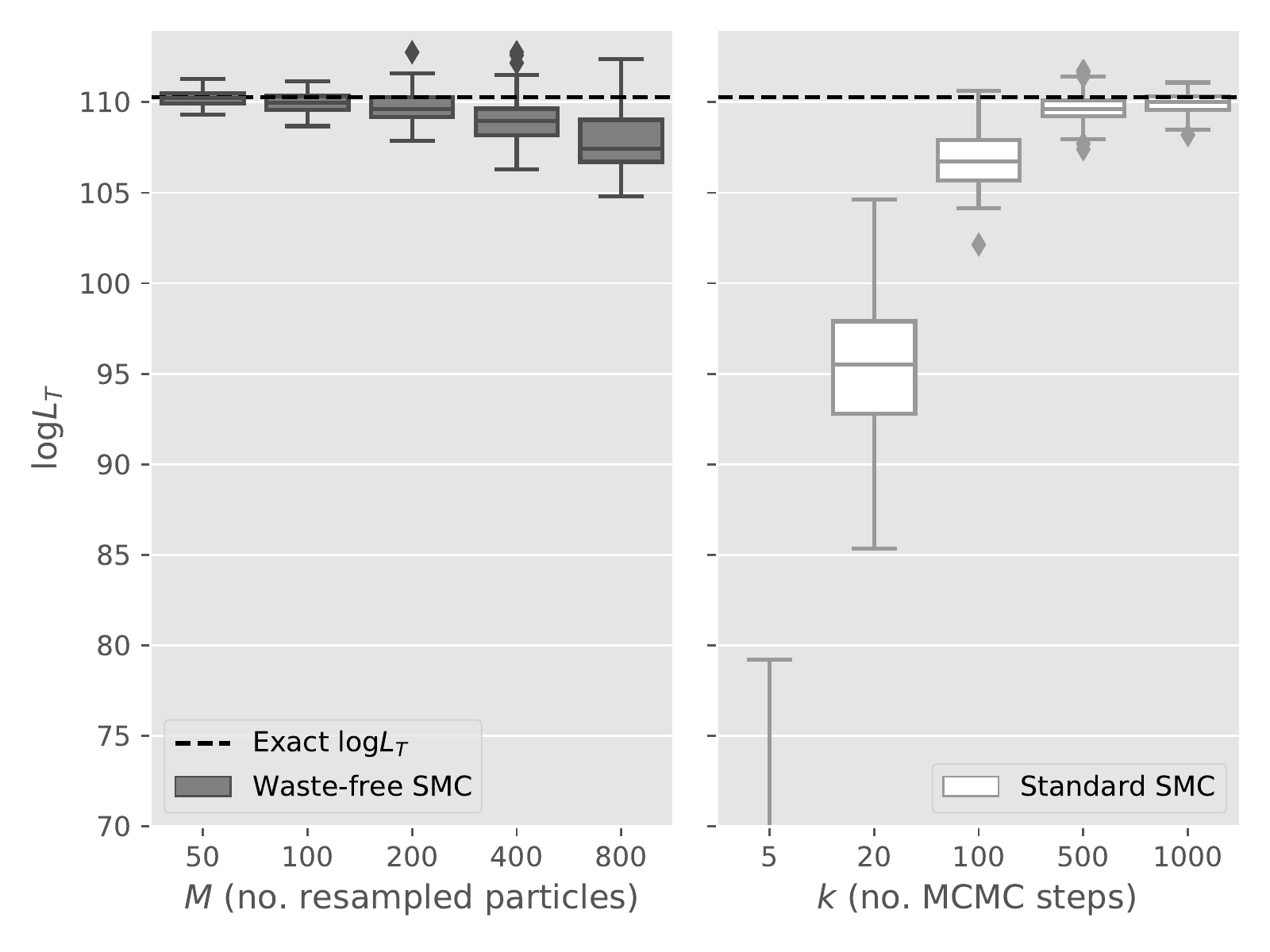}
  \caption{Latin squares: box-plots of estimates of $\log L_T$ (log of number of
    Latin squares) obtained from 100 independent runs of the following
    algorithms: waste-free SMC ($N=2\times 10^5$, different values of $M$, the number of
    resampled particles), and standard SMC ($N=2 \times 10^5/k$, different values for
    $k$, the number of MCMC steps).
    \label{fig:latin_logLT}
  }
\end{figure}

As in the previous example, the compared algorithms are given (roughly) the same
CPU budget: $N=2 \times 10^5/k$ for standard SMC, while $N=2 \times 10^5$ for waste-free (and
$k=1$, as already discussed). We make the same observations as in the previous
example: best performance is obtained from waste-free SMC in the long chain
regime ($M \ll N$), and, within this regime, performance does not seem to depend
strongly on $M$.

One distinctive feature of this example is that the mixing of the Metropolis
kernel used to move the particles significantly decreases over time; see
Figure~\ref{fig:latin_ar_vs_t}, which plots the acceptance rate of that kernel
at each iteration $t$ of a waste-free SMC run.

\begin{figure}
  \centering \includegraphics[scale=0.4]{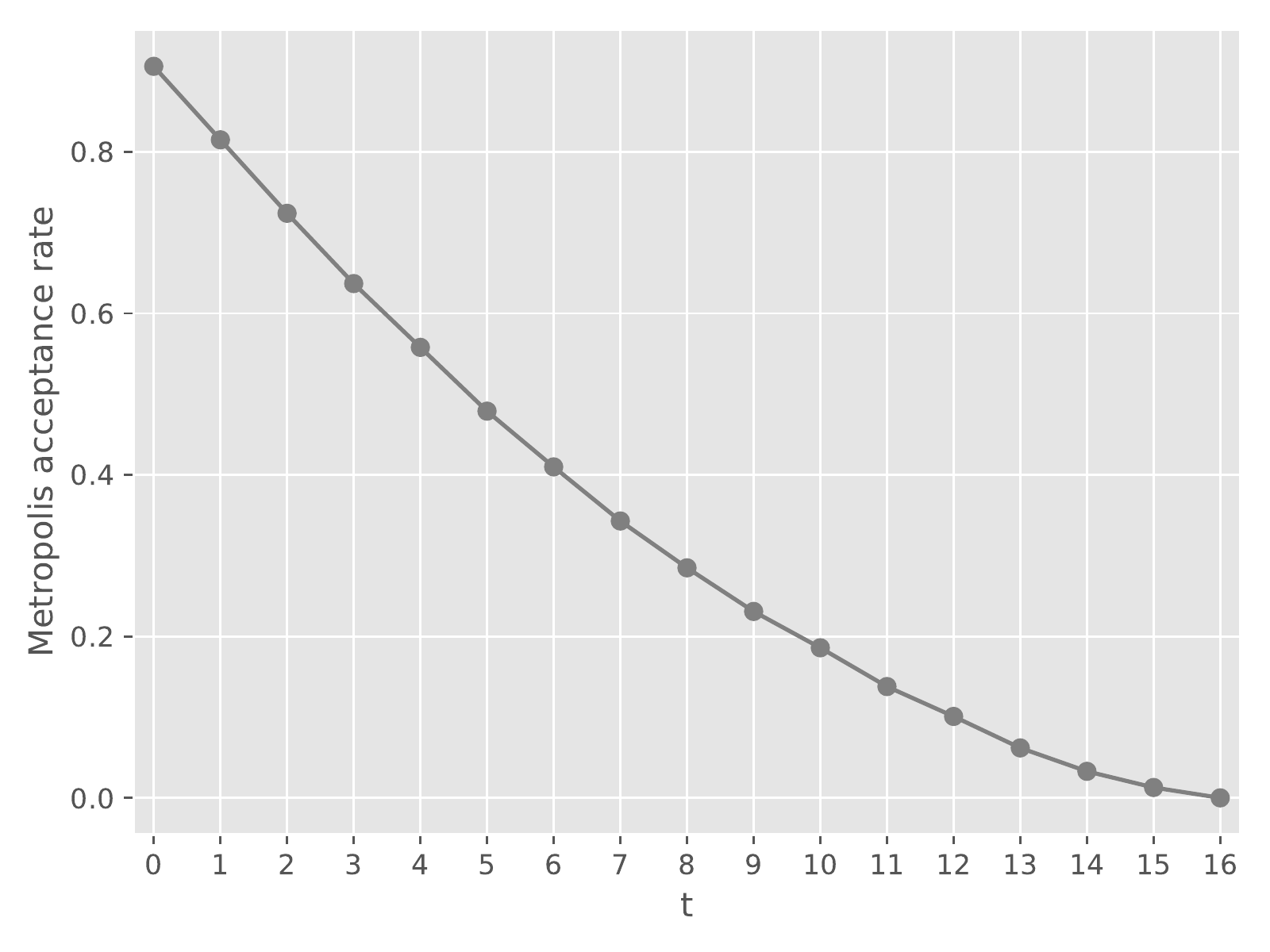}
  \caption{Latin squares: acceptance rate of the Metropolis kernel described in
    the text at each iteration $t$ of a run of waste-free SMC.
    \label{fig:latin_ar_vs_t}
  }
\end{figure}

It is interesting to note that waste-free SMC seems to work well despite this.
Unfortunately, it does seem to affect the performance of our $M-$chain variance
estimate. The left panel of Figure~\ref{fig:latin_var} makes the same
comparison as Figure~\ref{fig:bin_sonar_var} in our first example. This time,
however, the $M-$chain estimator seems to be biased downward, by a factor of
two. This bias seems to originate from the terms of for the last values of $t$;
these terms are both larger, and more difficult to estimate if $P$ is not large
enough.

\begin{figure}
  \centering \includegraphics[scale=0.37]{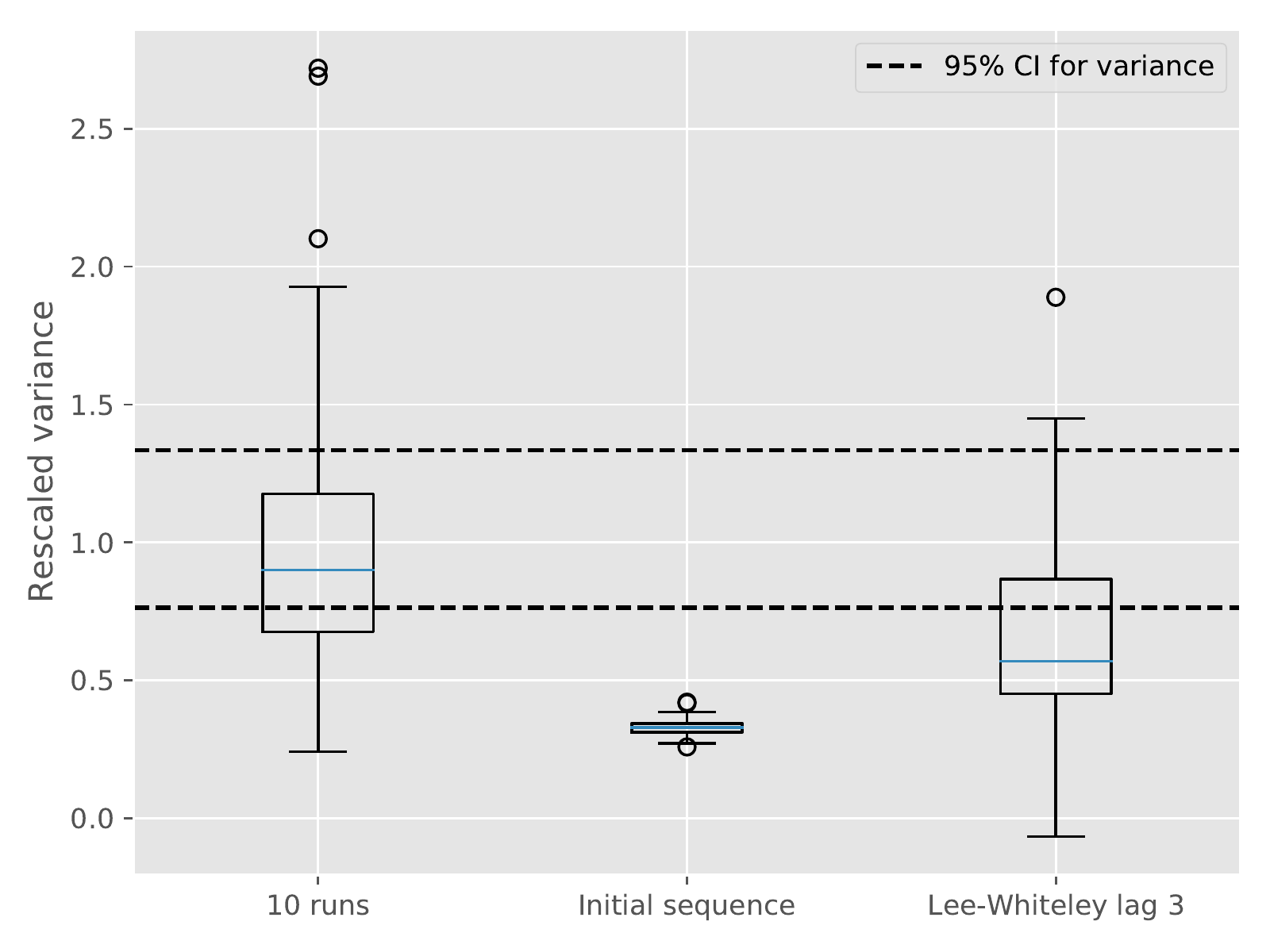}
  \includegraphics[scale=0.37]{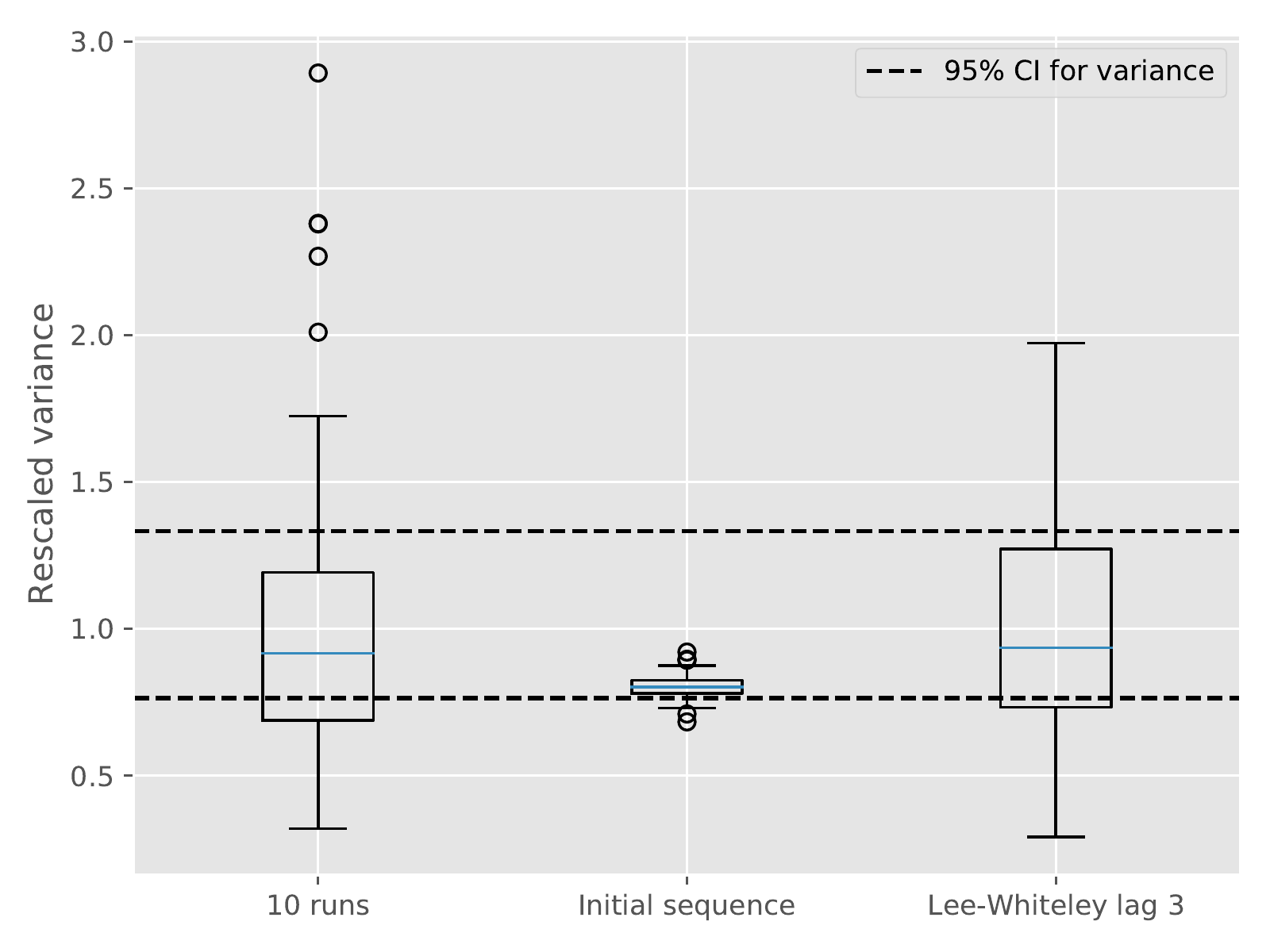}
  \caption{Latin squares: same plot as Figure~\ref{fig:bin_sonar_var}, for the 
    estimate of log-normalising constant $\log L_T$. Left: non-adaptive version ($M=50$,
    $N=2 \times 10^5$); Right: adaptive version ($M=50$, $N_0=5$). See text for more
    details. 
    \label{fig:latin_var}
  }
\end{figure}

These results showcase the interest of adapting $P$ across time, as discussed in
Section~\ref{sub:adaptive}. We re-run waste-free SMC for the same problem, with
$M=50$, and $\kappa=5$; that is, at each iteration $t$, $P_t$ is adjusted to be
close to $\kappa$ times the auto-correlation time, for function $\bar{G}_t$. The
right side of Figure~\ref{fig:latin_var} repeats the comparison of the variance
estimates, but for the adaptive $P$ algorithm. This time, our $M-$chain estimate
seems to perform satisfactorily.

In addition, Figure~\ref{fig:latin_perf_adaptive} compares the CPU vs error
trade-off for both variants of waste-free SMC. In both cases, we set $M=50$;
``CPU time'' on the x-axis is measured by the number of calls to the score
function, re-scaled so that the smallest observed value is 1. (Both axis use a
$\log2$-scale.) Each dot corresponds to an average over 100 runs. For the
vanilla version, we set $N=6250$, $2.5\times 10^4$, $10^5$, $4\times 10^5$ and
$8\times 10^5$. For the adaptive version, we set $\kappa=2$, 5 and 10. The
dotted lines have slope $-1$. For high CPU time both algorithms show the same
level of performance. If $N$ is set to too low a value for vanilla waste-free
(e.g. $N=6250$), then one obtains a very large MSE, because $P=N/M=125$ is too
small relative to the auto-correlation time of the kernels $M_t$ for large $t$.
Note that for the adaptive version, it does not make sense to take $\kappa\ll 2$,
as one cannot properly estimate the auto-correlation time of a chain without
running it for a length commensurate with its auto-correlation time. In a sense, 
the adaptive version of waste-free prevents us from setting $P$ to too low
a value, where performance becomes sub-optimal.

\begin{figure}
  \centering \includegraphics[scale=0.37]{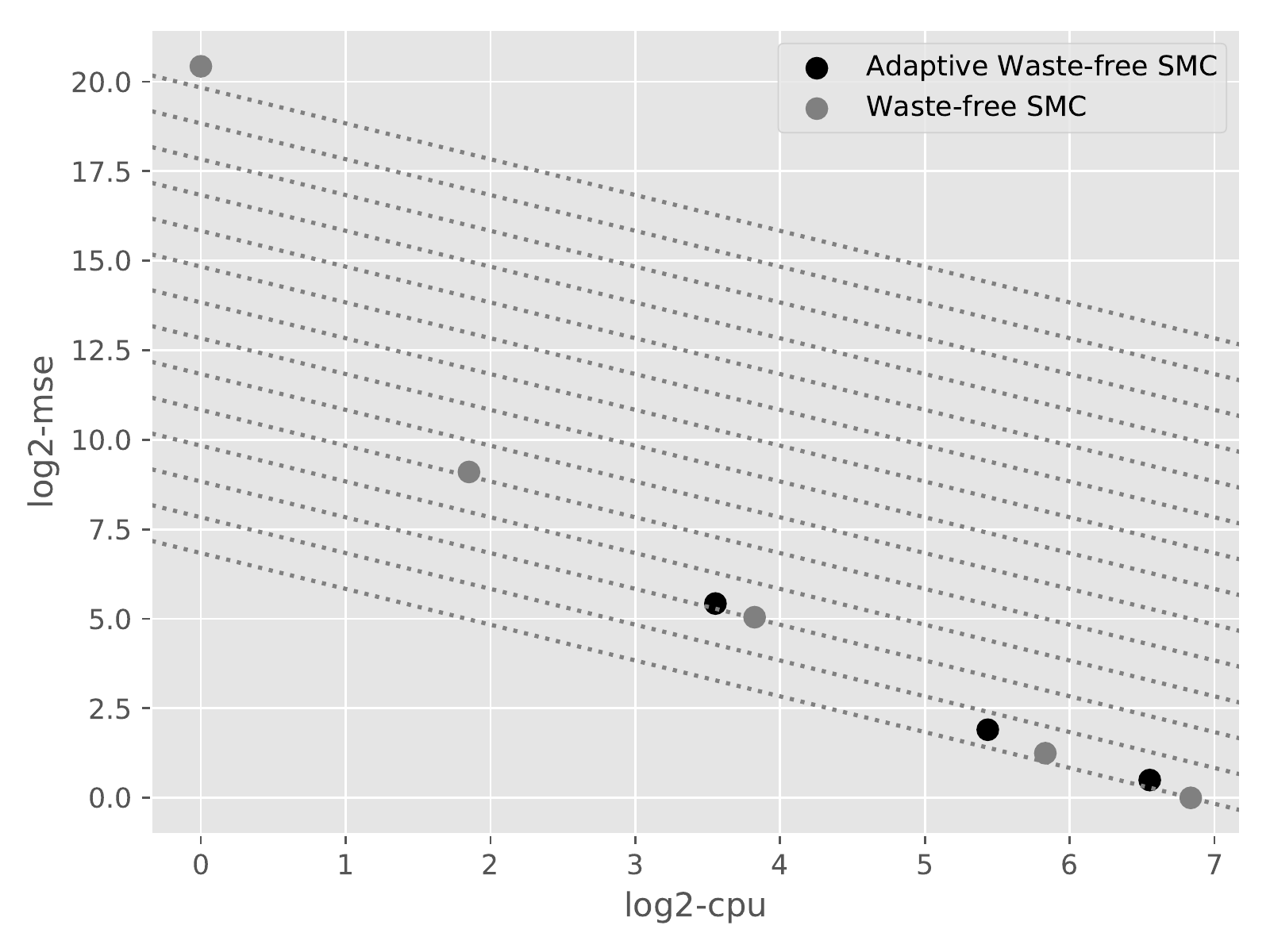}
  \caption{Latin squares: MSE (mean square error) vs CPU time (number of calls
    to score function), averaged over 100 independent runs of vanilla waste-free
    (grey dots, $N=6250$, $2.5\times 10^4$, $10^5$, $4\times 10^5$, $8\times
    10^5$), and adaptive waste-free (black dots, $\kappa=2$, $5$, $10$).
    Both axes use a $\log2-$scale; parallel dotted lines have slope $-1$.}
  \label{fig:latin_perf_adaptive}
\end{figure}

By and large, in any problem when there is some evidence that the mixing of
kernels $M_t$ may decrease significantly over time, we recommend to use the
adaptive $P$ strategy. It is a bit less practical to use, as it gives less
control to the user on the running time of the algorithm; but on the other hand
it seems to provide more reliable  variance estimates in this kind of scenario.

\subsection{Orthant probabilities}
\label{sub:orthant}

Finally, we consider the problem of evaluating Gaussian orthant probabilities,
i.e. $p(a,\Sigma)\eqdef\P(Z\geq a)$, where $a\in\R^d$, $Z\sim\N_d(0, \Sigma)$,
and $\Sigma$ is a covariance matrix of size $d\times d$.

\cite{MR3515028} developed the following SMC approach for evaluating such
probabilities. Let $\Gamma$ be the lower triangle in the Cholesky decomposition
of $\Sigma$: $\Sigma=\Gamma \Gamma^T$; $\Gamma=(\gamma_{ij})$ and
$\gamma_{ii}>0$ for all $i$. The orthant probability $p(a,\Sigma)$ may be
rewritten as the joint probability that $X_t \geq f_t(X_{1:t-1})$ for
$t=1,\ldots,d$, where $f_t(x_{1:t-1}) = (a_t-\sum_{s<t} \gamma_{st}x_s) /
\gamma_{tt} $, and the $X_t$'s are IID $\N(0,1)$ variables. (At time $1$,
$f_1(x_{1:0})$ is simply $a_1$, i.e. the constraint is $X_1\geq a_1$.)

The SMC algorithm of \cite{MR3515028} applies the following operations to
particles $X_{1:t}^n$, from time $1$ to time $T=d$. (We change notations
slightly and start at time 1, for the sake of readability.) (a) At time $t$,
particles $X_{1:t-1}^n$ are extended by sampling an extra component, $X_t^n$,
from a univariate truncated Gaussian distribution (the distribution of $X_t\sim
N(0,1)$ conditional on $X_t\geq f_t(x_{0:t-1})$); (b) particles $X_{1:t}^n$ are
then reweighted according to function $\Phi(-f_t(X_{1:t-1}^n))$, where $\Phi$ is
the $\N(0, 1)$ cumulative distribution function; and (c) when the ESS (effective
sample size) of the weights gets too low, the particles are moved through $k$
iterations of a certain MCMC kernel that leaves invariant $\pi_t$, the
distribution that corresponds to $X_{1:t}\sim N_t(0,I_t)$ constrained to
$X_s\geq f_s(X_{1:s-1})$ for $s=1,\ldots,t$. Based on numerical experiments,
\cite{MR3515028} recommended to use for the MCMC kernel at time $t$ a Gibbs
sampler that leaves $\pi_t$ invariant. (the update of each variable amounts to
sampling from a univariate truncated normal distribution.)

This SMC algorithm does not fit in the framework of
Algorithm~\ref{alg:genericSMC}; in particular the dimension of the state-space
$\setX=\R^t$ increases over time. However, we can easily generalise waste-free
SMC to this setting: whenever an MCMC rejuvenation step is applied, resample
$M\ll N$ particles, apply $P-1$ steps of the chosen MCMC kernels to these $M$
resampled particles, and gather the $N=MP$ so obtained values to form the new
particle sample.

To make the problem challenging, we take $d=150$, $a=(1.5, 1.5, \ldots)$, and
$\Sigma$ a random correlation matrix with eigenvalues uniformly distributed in
the simplex $\left\{  x_1 + \cdots + x_d = 150, x_i \geq 0\right\}$, which we
simulated using the algorithm of \citet{davies2000numerically}. As in
\cite{MR3515028}, before the computation we re-order the variables according to
the heuristic of \citet{gibson1994monte}. 

Figures~\ref{fig:orthant_logLT} and \ref{fig:orthant_QT} do the same comparison
of standard SMC and waste-free SMC as in the two previous examples: $N=2\times
10^5$ for waste-free, $N=2\times10^5/k$ for standard SMC, and $M$ (resp. $k$)
varies over a range of values. Figure~\ref{fig:orthant_logLT} plots box-plots
of estimates of $\log L_T$ (the log of the orthant probability), while
Figure~\ref{fig:orthant_QT} does the same for $\Q_T(\varphi)$, with
$\varphi(x_{0:T})= (\sum_{t=0}^Tx_t)/T$; i.e. the expectation of $\varphi$ with
respect to the corresponding truncated Gaussian distribution. 

We observe again that waste-free SMC outperforms standard SMC, at least whenever
$M\ll N$. In addition, the greater robustness of waste-free is quite striking 
in this example.

\begin{figure}
  \centering \includegraphics[scale=0.55]{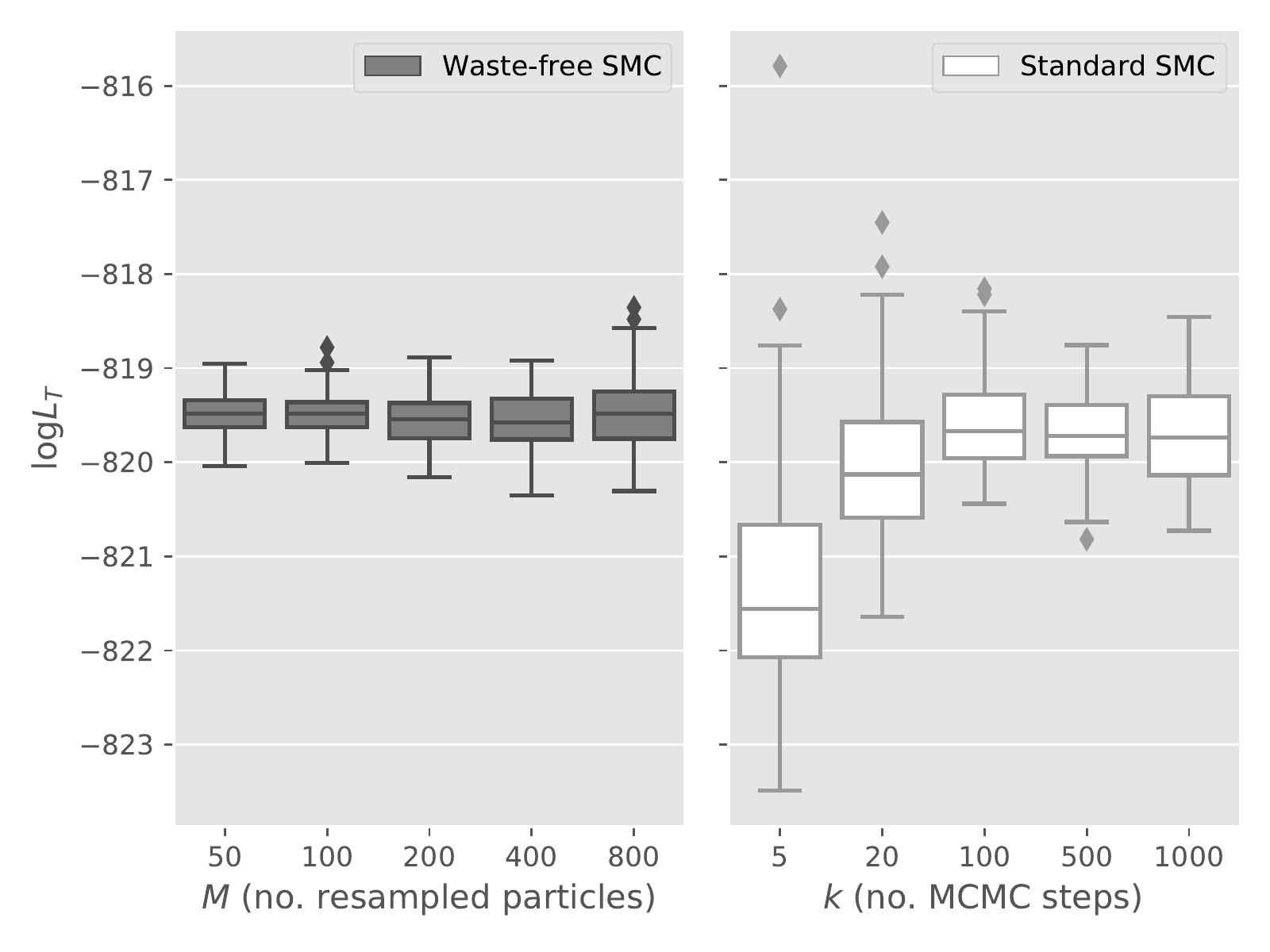}
  \caption{Orthants: estimates of the log normalizing constant obtained
    from waste-free SMC ($N=2 \times 10^5$) and standard SMC ($N=2\times10^5/k$). 
    \label{fig:orthant_logLT}
  }
\end{figure}

\begin{figure}
  \centering \includegraphics[scale=0.55]{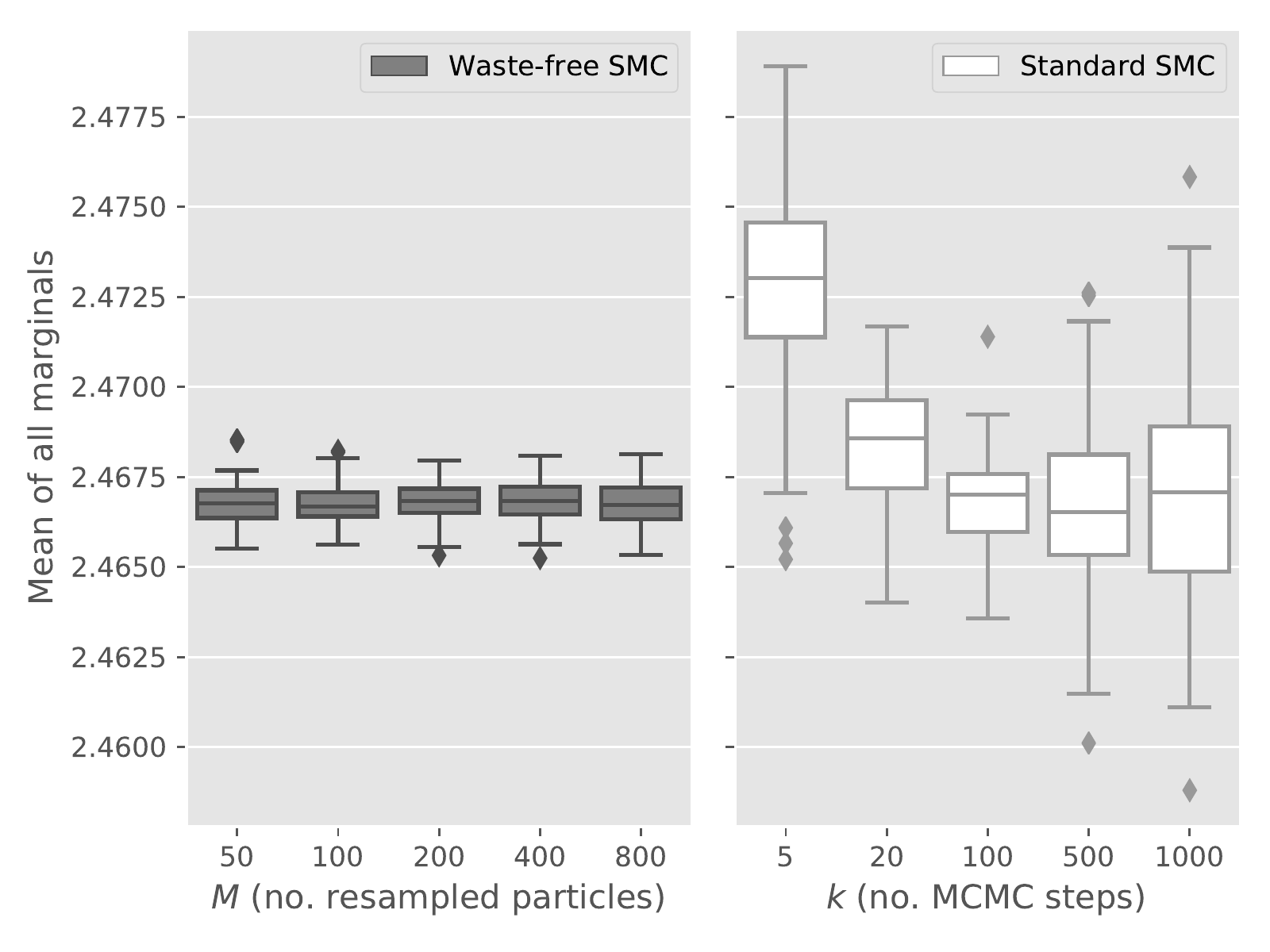}
  \caption{Orthants: Same plot as Figure~\ref{fig:orthant_logLT} for
    $\Q_T(\varphi)$, the expectation of function $\varphi(x_{0:T})=(\sum_{t=0}^T
    x_t)/T$ with respect to truncated Gaussian distribution $\N_{>0}(a,\Sigma)$. 
    \label{fig:orthant_QT}
  }
\end{figure}

Finally, Figure~\ref{fig:orthant_var} compares $M-$chain estimators of the
variance of the orthant probability estimate based on two single-chain
estimators: the initial sequence estimator we recommended by default in Section
\ref{sub:variance_est}, and we used in the two previous examples; and a spectral
estimator based on the Tukey-Hanning window \citep[see e.g.][]{MR2604704}. In
this example, the kernels $M_t$ are Gibbs kernels, and are therefore
not reversible. This seems to explain the poor performance of the former.

(As in previous plots, Figures~\ref{fig:bin_sonar_var} and \ref{fig:latin_var},
we include for comparison the variance estimator obtained by taking an empirical
variance over 10 runs; however we do not include, for the sake of readability,
the estimator based on \cite{Lee2018}, but note simply it performs poorly in
this case too.)

\begin{figure}
  \centering
  \includegraphics[scale=0.35]{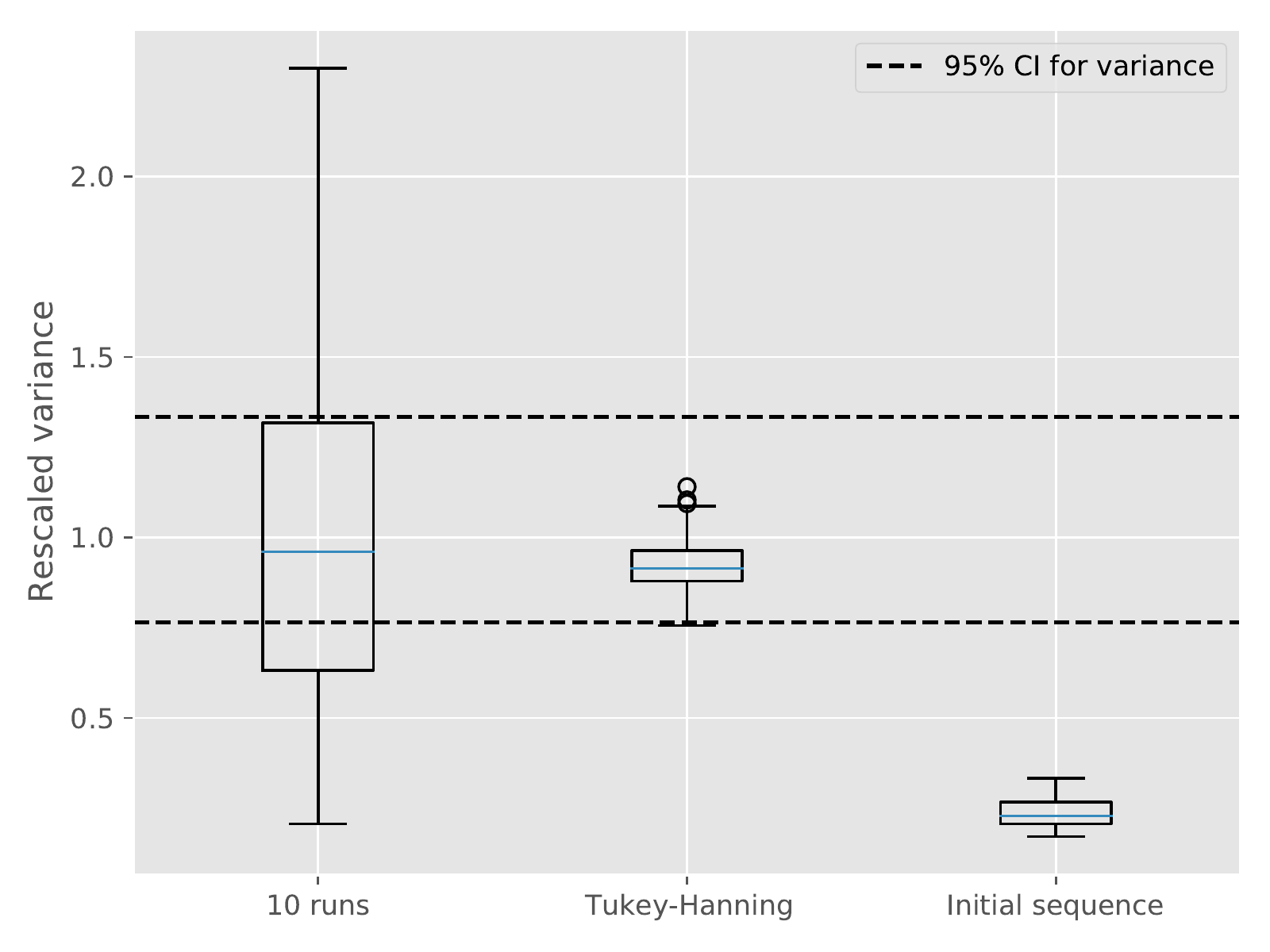}
  \caption{Orthants: box-plots of variance estimates over 100 runs for the
estimate of the log orthant probability. The variance estimates are re-scaled so
that the empirical variance over the 100 runs equals one; see text for more
details.
   \label{fig:orthant_var} 
 }
\end{figure}

\section{Concluding remarks}
\label{sec:conclusion}

\subsection{Connection with nested sampling}
\label{subsec:nested}

In our definition of waste-free SMC, we took $N=MP$, with $P\geq 2$; thus $M$
divides $N$. We may generalise the algorithm to any pair $(M, N)$, $M<N$: at
time $t$, resample $M$ particles, generate $M$ chains of length $k\eqdef \lfloor
N /M \rfloor $ (using kernel $M_t$, and the resampled particles as the starting
points); then select (without replacement) $N-Mk$ chains and extend them to have
length $k+1$. The total number of particles is then $N$.

One interesting special case is $M=N-1$. In that case, $N-1$ particles are
resampled (thus at least one particle is discarded), and, among these $N-1$
resampled particles, only one particle is moved through kernel $M_t$. In addition,
if the target distributions $\pi_t$ are of the form $\pi_t(\dx)\propto \nu(\dx)
\ind\{ L(x) \geq l_t\}$, where $\nu$ is a prior distribution, and $L$ a
likelihood function, then one recovers essentially the nested sampling algorithm
of \cite{MR2282208}.

This raises the question whether the regime $M=N-1$ is useful, either for such a
sequence of distributions, or more generally. For the former, the numerical
experiments of \cite{salomone2018unbiased} seem to indicate than standard SMC,
when applied to this type of sequence, may perform as well as nested sampling.
This suggests waste-free SMC should also perform at least as well as nested
sampling, although we leave that point for further investigation. For the
latter, we note that taking $M=N-1$ is not very convenient, as this means we
move only one particle at each iteration, although each iteration costs
$\bigO(N)$. (In nested sampling, the cost of a single iteration may be reduced
to $\bigO(1)$ by using the fact that weights are either 0 or 1.)

\subsection{Further work}
\label{subsec:label}

Our convergence results assume that the kernels $M_t$ are uniformly ergodic.
However, many practical MCMC kernels are not uniformly ergodic, hence it seems
worthwhile to extend these results to, say, geometrically ergodic kernels. 
Another result we would like to establish is that waste-free SMC dominates
standard SMC in terms of asymptotic variance, at least under certain conditions
on the mixing of the kernels $M_t$.

In terms of applications, we wish to explore how waste-free may be implemented
in various SMC schemes, in particular in the SMC$^2$ algorithm of
\cite{MR3065473}. This algorithm is an SMC sampler with expensive Markov kernels
(as a single step amounts to propagate a large number of particles in a
``local'' particle filter), hence the benefits brought by waste-free SMC may be
particularly valuable in this type of scenario. 

The original implementation of the numerical examples may be found at
\url{https://github.com/hai-dang-dau/waste-free-smc}.  Waste-free SMC is also
now implemented in the \verb+particles+ library, see 
\url{https://github.com/nchopin/particles}.

\section*{Acknowledgements}

The first author acknowledges a CREST PhD scholarship via AMX funding. The
second author acknowledges partial support from Labex Ecodec
(Ecodec/ANR-11-LABX-0047).  We are grateful to Chris Drovandi, Pierre Jacob and
two anonymous referees for helpful comments on a preliminary version of the
paper. 

\appendix
\section{Proofs}

\subsection{Proof of Proposition~\ref{prop:unbiased}}

We may rewrite \eqref{eq:norm_const_est} as:
\[
    \prod_{s=0}^t \left\{ \frac 1 M \sum_{m=1}^M \left( \frac 1 P \sum_{p=1}^P
    G_s(z_s^m[p]) \right) \right\}
\]
where $z_s^m[p]$ stands for variable $\tilde{X}_s^{m,p}$ which is defined
inside Algorithm~\ref{alg:wasteless}.

We recognise the normalising constant estimate of a standard SMC sampler,
Algorithm~\ref{alg:genericSMC}, when applied to the waste-free \FK~model defined
in Proposition~\ref{prop:wf_FK}. The expectation of this quantity is therefore
the normalising constant $L_t^{\wf}=L_t$ (Proposition~\ref{prop:wf_FK}), since
such estimates are unbiased \citep{DelMoral1996unbiased}.

\subsection{Proof of Proposition~\ref{thm:L2}}
\label{sub:proof_thL2}

We start by establishing two technical lemmas regarding a uniformly ergodic
Markov chain $(X_p)_{p\geq 0}$, $(X_p)$ for short, on probability space
$(\setX, \mathbb{X})$; i.e. $\tv{K^k(x,\dx') -  \pi(\dd x')} \leq C \rho^k$ for
certain constants $C>0$ and $\rho<1$ and a certain probability distribution
$\pi$, where $K^k(x, \dx)$ stands for the $k-$fold Markov kernel that defines
the distribution of $X_{p+k}$ given $X_p$. Then $\pi(\dx)$ is its stationary
distribution.

\begin{lem}
    \label{lem:unif_ergo_finite_variance} Assume that $(X_p)$ is stationary,
    i.e. $X_0\sim \pi(\dx)$, and therefore
    $X_p \sim \pi(\dx)$ for all $p\geq 0$. Then there exists a constant $C_1>0$
    such that:
    \begin{equation*}
        \Var\left(\varphi( X_0)\right)
        + 2 \sum_{k=1}^\infty \left|\Cov\left(\varphi(X_0), \varphi(X_k)\right)\right|
        \leq C_1 \infnorm{\varphi}^2
    \end{equation*}
    for any measurable bounded function $\varphi:\setX\rightarrow \R$.
\end{lem}

\begin{proof}
    One has
    \begin{align*}
        \left|\Cov\left(\varphi (X_0), \varphi (X_k)\right)\right|
    &= \left|\E[\varphi (X_0) \varphi (X_k)] - \E[\varphi X_0] \E[\varphi X_k]\right| \\
    &= \left|\int \left\{\int \varphi(x_k) K^k(x_0,\dx_k) - \int \varphi(x_k)
    \pi(\dx_k)  \right\} \varphi(x_0)\pi(\dx_0) \right| \\
    &\leq 2 \rho^k C \infnorm{\varphi}^2 
    \end{align*}
    from which the result follows.
\end{proof}

In the second lemma, the distribution of the initial state $X_0$ is arbitrary,
and therefore the chain is not necessarily stationary.

\begin{lem} \label{lem:var_unif_ergo} 
    There exists a constant $C_2>0$ (which does not depend on the initial
    distribution of the chain, i.e. the distribution of $X_0$), such that
    \begin{equation*}
        \Var\left(\frac 1 P \sum_{p=1}^P \varphi(X_p)\right) 
        \leq C_2 \frac{\infnorm{\varphi}^2}{P}
    \end{equation*}
    for any $P\geq 1$ and any bounded measurable function
    $\varphi:\setX\rightarrow \R$.
\end{lem}

\begin{proof}
  The proof relies on a standard coupling argument, see e.g. Chapter 19 of
  \cite{MR3889011}. We introduce an arbitrary integer $R$, $1 \leq R \leq P$,
  and a Markov chain $(X_p^\star)$ constructed as follows: (a) $X_0^\star\sim
  \pi(\dx)$, the stationary distribution of $(X_p)$; (b) variables $X_R$,
  $X_R^\star$ are maximally coupled, which implies that:
  \begin{equation}
    \label{eq:coupling_prob}
    \P(X_R \neq X_R^\star) = \tv{\int \mu(\dx_0) K(x_0,\dx_p) - \pi(\dx_p)}  \leq C \rho^R
  \end{equation}
  where $\mu(\dx_0)$ denotes the probability distribution of $X_0$, and the
  inequality stems from the uniform ergodicity of the chain; (c) if
  $X_R=X_R^\star$, the two chains remain equal until time $P$, otherwise they
  are independent; (d) the distribution of $X_1^\star,\ldots,X_{R-1}^\star$
  given $X_0^\star$, $X_R^\star$ is the conditional distribution of these states
  induced by $K(x, \dx')$, the Markov kernel of $(X_p)$. For more details on
  maximal coupling of two probability distributions, see e.g. Chap. 19 of 
  \cite{MR3889011}.
  
  Using the inequality $\Var(X + Y + Z) \leq 3 (\Var(X) + \Var(Y) + \Var(Z))$,
  we have:
  \begin{align*}
    \Var\left( \frac 1 P \sum_{p=1}^P \varphi(X_p) \right)
    & \leq 3 \Var\left( \frac 1 P \sum_{p=1}^P \varphi(X_p^\star) \right) 
      + 3 \Var\left( \frac 1 P \sum_{p=1}^{R} \left\{\varphi(X_p)
      - \varphi(X_p^\star) \right\} \right) \\
    & + 3 \Var\left( \ind\{X_R\neq X_R^\star\}
      \frac 1 P \sum_{p=R}^P \left\{\varphi(X_p) - \varphi(X_p^\star) \right\} \right) \\
    & \leq  \frac{ 3 C_1\infnorm{\varphi}^2}{P}
      + \frac{12 R^2 \infnorm{\varphi}^2 }{P}
      + C \rho^R \infnorm{\varphi}^2 
  \end{align*}
  where we have applied Lemma~\ref{lem:unif_ergo_finite_variance} to the first
  term, and \eqref{eq:coupling_prob} to the third term. We conclude by taking $R
  = \lceil \sqrt{P} \rceil$.

\end{proof}

We now prove Proposition~\ref{thm:L2} by induction. Clearly, \eqref{eq:L2_pred}
holds at time $0$. The implication \eqref{eq:L2_pred} $\Rightarrow$
\eqref{eq:L2_filt} at time $t$ follows the same lines as for a standard SMC
sampler, see e.g. Section 11.2.2 in \cite{SMCbook}. Now assume that
\eqref{eq:L2_filt} holds at time $t-1\geq 0$, and let
$\bar{\varphi}=\varphi-\Q_{t-1}(\varphi)$,
$\mathcal{F}_{t-1}=\sigma(X_{t-1}^{1:N})$ (the $\sigma$-field generated by
variables $X_{t-1}^n$, $n=1,\ldots,N$). Then
\begin{flalign*}
  \MoveEqLeft \CE{\pr{\frac 1 N \sum_{n=1}^N \varphi(X_t^n) - \Q_{t-1}
      (\varphi)}^2}{\mathcal{F}_{t-1}} \\
  & = \CE{\pr{\frac 1 M \sum_{m=1}^M \frac 1 P
      \sum_{p=1}^P \bar{\varphi}(\tilde{X}_t^{m,p}) }^2}{\mathcal{F}_{t-1}} \\
  &= \pr{\CE{\frac 1 P \sum_{p=1}^P
      \bar{\varphi}(\tilde{X}_t^{1,p})}{\mathcal{F}_{t-1}} }^2 + \frac 1 M
  \CVar{\frac 1 P \sum_{p=1}^P \bar{\varphi}(X_t^{1,p})}{\mathcal{F}_{t-1}}
\end{flalign*}
since the blocks of variables $X_t^{m,1:P}$ are IID (independent and identically
distributed) conditional on $\mathcal{F}_{t-1}$.
 
The expectation of the first term may be bounded by $c'_{t-1}
\infnorm{\varphi}^2/N$ by applying \eqref{eq:L2_filt} to function $ P^{-1}
\sum_{p=0}^{P-1} M_t^p \varphi$. The second term may be bounded by $C_2
\infnorm{\varphi}^2 /N$ using Lemma~\ref{lem:var_unif_ergo}.

\subsection{Proof of Theorem~\ref{thm:clt_longchain}}
\label{sec:proof_thm_clt_long}

We start by proving a few basic lemmas. The first one concerns product measures.
We use symbol $\bigotimes$ throughout to represent the product of two
probability measures.

\begin{lem}[Total variation distance for product measure]
  \label{lem:tv_prod_mes}
  Let $\mu_{1:N}$ and $\nu_{1:N}$ be $2N$ probability measures on $(\setX,
  \mathbb{X})$. Then the following inequality holds:
$$ \tv{\bigotimes_{n=1}^N \mu_n-\bigotimes_{n=1}^N \nu_n}
\leq \sum_{n=1}^N \tv{\mu_n- \nu_n}. $$
\end{lem}
\begin{proof}
  Take $N=2$. Then
  \[ \tv{\mu_1 \otimes \mu_2 - \nu_1 \otimes \nu_2} \leq \tv{\mu_1 \otimes \mu_2
      - \mu_1 \otimes \nu_2} + \tv{\mu_1 \otimes \nu_2 - \nu_1 \otimes \nu_2}
  \]
  and we may bound the first term as follows:
  \begin{flalign*}
    \MoveEqLeft \tv{\mu_1 \otimes \mu_2 - \mu_1 \otimes \nu_2} \\
    & = \sup_{f: \setX^2 \to [0,1]} \abs{\int \left( \int f(x,y)\mu_1
        (\dx)\right)
      \mu_2 (\dd y) - \int \pr{\int f(x,y)\mu_1 (\dx)} \nu_2 (\dd y)} \\
    & \leq \sup_{g: \setX \to [0,1]} \abs{\int g(y) \mu_2(\dd y) - \int g(y)
      \nu_2(\dd y)} = \tv{\mu_2 - \nu_2}.
  \end{flalign*}
  The result follows by bounding the second term similarly. For $N\geq 3$,
  proceed recursively.
\end{proof}

The two next lemmas concern the behaviour of $M\geq 1$ independent, stationary,
Markov chains, $(Y_p^m)_{p\geq 0}$ on $(\setX, \mathbb{X})$, $m=1,\ldots,M$ with
uniformly ergodic Markov kernel $K$, and invariant distribution $\pi$:
$\tv{\delta_x K^p - \pi} \leq C \rho^k$ for constants $C\geq 0$ and
$\rho\in[0,1)$.

\begin{lem}
  \label{lem:product_kernel}
  The product kernel
  \begin{equation*}
    K^{\otimes M} (x_{1:M}, \dx'_{1:M}) = \prod_{m=1}^M K(x_m, \dx'_m)
  \end{equation*}
  is uniformly ergodic, with stationary distribution $\pi^{\otimes M}$.
\end{lem}

\begin{proof}
  This is a direct consequence of Lemma~\ref{lem:tv_prod_mes}:
  \begin{align*}
    \tv{\delta_{x_{1:M}} \left( K^{\otimes M} \right)^p - \pi^{\otimes M}}
    & \leq \sum_{m=1}^M \tv{\delta_{x_m}K^p -\pi} \\
    & \leq C M \rho^p. 
  \end{align*}
\end{proof}

\begin{lem}
  \label{lem:clt_Mchains}
  For $\varphi:\setX \rightarrow \R$ measurable and bounded, one has:
  \[
    \sqrt{MP} \left( \frac{\sum_{m=1}^M\sum_{p=1}^P \varphi(Y_p^m) }{MP} -
      \pi(\varphi) \right) \cvd \N\left( 0, \vas(K, \varphi) \right)
  \]
  as $P\rightarrow +\infty$, whether $M\geq 1$ is fixed, or $M$ grows with $P$;
  i.e. $M=M(P)\rightarrow +\infty$ as $P\rightarrow +\infty$.
\end{lem}
\begin{proof}
  For $M=1$, this is simply the classical central limit theorem for uniformly
  ergodic Markov chains, see e.g. Theorem 23 in \cite{roberts2004general} and
  references therein. For $M\geq 2$ fixed, we may apply the same theorem to the
  Markov chain $(Y_p^{1:M})_p$ in $(\setX^M,\mathbb{X}^M)$, which is also
  uniformly ergodic (Lemma~\ref{lem:product_kernel}) and to test function
  $\varphi_M(y^{1:M})=M^{-1}\sum_{m=1}^M\varphi(y^m)$.

  Assume now $M=M(P)$ grows with $P$. Let
  $\bar{\varphi}=\varphi-\pi(\varphi)$ and let $S_P$ denote a variable with
  the same distribution as $S_P^m \eqdef P^{-1/2}\sum_{p=1}^P
  \bar{\varphi}(Y_p^m)$ for $m=1,\ldots, M$. (These $M$ variables are IID.) By the formula (19) of \cite{roberts2004general}, we have $\E[S_P^2] \rightarrow v_\infty(K,\varphi)$. Therefore, 
 fixing $u\in \R$, we wish to prove that $\Delta_P\rightarrow 0$, where
  \begin{equation*}
    \Delta_P := \abs{\pr{\E e^{iuS_P / \sqrt M}}^M - \pr{1 - \frac{u^2}{2M} \E(S_P^2)}^M}. 
  \end{equation*}

  Let $M_0 \geq 1$ be fixed such that $u^2 \E[S_P^2]/2M_0<1$ for all $P > M_0$. Since $\abs{a^M-b^M}\leq M\abs{a-b}$ for $\abs{a}, \abs{b} \leq 1$ and
  $\abs{e^{ix} - 1 - ix + x^2/2} \leq \min(x^2, |x^3|/6)$ for $x \in
  \mathbb{R}$, we have, for any $M \geq M_0$:
  \begin{equation}
    \label{eq:deltaP}
    \Delta_P
    \leq \E \min \pr{u^2 S_P^2, \frac{\abs{u^3 S_P^3}}{6 \sqrt M}}
    \leq \E f_{M_0} (S_P)
  \end{equation}
  where $f_m(x) := \min\pr{u^2 x^2, \abs{u^3 x^3 / 6 \sqrt{m}}} = f_m^1(x) +
  f_m^2(x)$, $f_m^1(x) := u^2 x^2$ and $f_m^2(x):= \ind_{\abs{x} \leq 6 \sqrt m
    / \abs{u}} \pr{\abs{u^3 x^3}/6 \sqrt{m} - u^2x^2}$. Then, if $G$ is a
  Gaussian variable with variance $v_\infty(K,\varphi)$, we have $\E f^1_{M_0} (S_P)
  \to \E f^1_{M_0} (G)$ as $P\rightarrow +\infty$. Moreover, $\E f^2_{M_0} (S_P) \to \E f^2_{M_0}(G)$ by
  Theorem 23 of \cite{roberts2004general} and the fact that $f^2_{M_0}$ is a
  bounded function and is only discontinuous on a set of measure zero with
  respect to a Gaussian distribution. Thus \eqref{eq:deltaP} implies that
  $\limsup_{P \to \infty} \Delta_P \leq \E f_{M_0}(G)$. But $\E f_{M_0}(G) \to
  0$ as $M_0 \to \infty$ by the dominated convergence theorem, hence $\Delta_P
  \to 0$ and the lemma is proved.
\end{proof}

We now prove Theorem~\ref{thm:clt_longchain}. We proceed by induction:
\eqref{eq:clt_pred} at time 0 is simply the standard central limit theorem for
IID variables. The implication \eqref{eq:clt_pred} $\Rightarrow$
\eqref{eq:clt_filt} at time $t$ may be established exactly as in other proofs
for central limit theorems for SMC algorithms; see e.g. Section 11.3 of
\cite{SMCbook}.

We now assume that \eqref{eq:clt_filt} holds at time $t-1\geq 0$, and we wish to
show that \eqref{eq:clt_pred} holds at time $t$, or, equivalently, that:
\begin{equation}
  \label{eq:clt_in_P}
  \frac{1}{\sqrt P} \sum_{p=1}^P \varphi_{M}(Z_p)
  \Rightarrow \N\left(0, v_\infty(M_t, \varphi)\right)
\end{equation}
where (dropping the dependence on $t$ as it is fixed) $Z_p\eqdef
(\tilde{X}_t^{1,p},\ldots,\tilde{X}_t^{M,p})$ is a Markov chain on $\setX^M$,
which is uniformly ergodic (Lemma~\ref{lem:product_kernel}), and
$\varphi_M(z)=M^{-1/2}\sum_{m=1}^M\bar{\varphi}(z[m])$.

We apply the coupling construction we used in the proof of Lemma
\ref{lem:var_unif_ergo} to this Markov chain: we introduce a stationary Markov
chain, $(Z_p^\star)$, with the same Markov kernel as $(Z_p)$, i.e. $M_t^{\otimes
  M}$, which is coupled to $(Z_p)$ at time $R$, $1\leq R \leq P$, with maximum
coupling probability:
\begin{equation}
  \label{eq:n_alpha_coupling_bound}
  \P(Z_R \neq Z_R^\star) =
  \tv{\mathcal{L}(Z_1) (M_t^{\otimes M})^R - \pi_{t-1}^{\otimes M}}
  \leq M C \rho^R
\end{equation}
If the two chains are successfully coupled at time $R$, they remain equal at
times $R+1, \ldots, P$.

We decompose the left-hand side of \eqref{eq:clt_in_P} as:
\begin{multline}
  \label{eq:decomp_clt}
  \frac{1}{\sqrt P} \sum_{p=1}^P \varphi_{M}(Z_p) = \frac{1}{\sqrt P}
  \sum_{p=1}^P \varphi_{M}(Z_p^\star)
  + \frac{1}{\sqrt P} \sum_{p=1}^R \varphi_{M}(Z_p)\\
  - \frac{1}{\sqrt P} \sum_{p=1}^R \varphi_{M}(Z_p^\star) + \frac{1}{\sqrt P}
  \ind\{Z_R\neq Z_R^\star\} \sum_{p=R+1}^P \left( \varphi_M(Z_p) -
    \varphi_M(Z_p^\star) \right).
\end{multline}

The first terms converges to $\N\left( 0,v_\infty(M_t,\varphi) \right)$, see Lemma
\ref{lem:clt_Mchains}. What remains to prove is that the three other terms
converge to zero in probability.

The fourth term is non-zero with probability \eqref{eq:n_alpha_coupling_bound},
and tends to zero as soon as $R\rightarrow +\infty$; e.g. $R=\bigO(P^\beta)$,
$\beta\in(0,1)$. Using the inequality $\Var(Y_1+\ldots+Y_R) \leq R \left(
  \Var(Y_1) + \ldots +\Var(Y_R) \right)$, we may bound the the $L^2$ norm of the
third term as follows:
\begin{equation*}
  \Var\left( \frac 1 {\sqrt{P}} \sum_{p=1}^R \varphi_M(Z_p^\star) \right)
  \leq \frac{R^2}{P} \Var_\pi(\bar{\varphi})
  \leq 2 \frac{R^2}{P} \infnorm{\varphi}^2
\end{equation*}
which tends to zero as soon as $R^2\ll P$, e.g. $R=\bigO(P^\beta)$,
$\beta\in(0,1/2)$.

The second term equals:
\begin{equation}
  \label{convg:n_alpha:2nd_term_bound}
  R  \sqrt{\frac{M}{P}} \pr{\frac 1 M \sum_{m=1}^M \frac 1 R
    \sum_{p=1}^R \bar{\varphi}(\tilde{X}_t^{m,p})}
\end{equation}
and, since the $M$ chains $\tilde{X}_t^{m,1:P}$ are independent, for
$m=1,\ldots,M$, conditional on $\mathcal{F}_{t-1}=\sigma(X_{t-1}^{1:N})$, we
have:
\begin{flalign*}
  \MoveEqLeft  \CE{\pr{\frac 1 M \sum_{m=1}^M \frac 1 R \sum_{p=1}^R \bar{\varphi}(X_{t}^{m,p}) }^2 }{\mathcal{F}_{t-1}} \\
  &= \pr{\CE{\frac 1 R \sum_{p=1}^R
      \bar{\varphi}(\tilde{X}_t^{1,p})}{\mathcal{F}_{t-1}}}^2
  + \frac 1 M \CVar{\frac 1 R \sum_{p=1}^R \varphi(\tilde{X}_{t}^{m,p})}{\mathcal{F}_{t-1}} \\
  &\leq \left\{\Q_{t-1}^N \pr{\frac 1 R \sum_{p=1}^R M_t^{p-1}\bar{\varphi}}
  \right\}^2 + \frac 2 M \infnorm{\varphi}^2
\end{flalign*}
where $\Q_{t-1}^N(\varphi)=\sum_{n=1}^NW_{t-1}^n \varphi(X_{t-1}^n)$.

The expectation of the first term can be bounded by a constant times
$\infnorm{\varphi}^2 /N$ by Proposition~\ref{thm:L2}, thus the $\mathbb{L}^2$ norm
of \eqref{convg:n_alpha:2nd_term_bound} is $\bigO(R/\sqrt{MP})$, which tends to
zero as soon $R^2\ll MP$. Taking $R=\bigO(P^\beta)$, $\beta\in(0, 1/2)$
therefore ensures that all the terms in \eqref{eq:decomp_clt}, minus the first,
goes to zero.

\subsection{Proof of Theorem~\ref{thm:clt_norm_cst}}
\label{subsec:proof_clt_norm_cst}

Before proving Theorem~\ref{thm:clt_norm_cst}, we need to define some new notations
to work comfortably with the convergence of conditional distributions. We start with
a simple example.

Most Markov chains used in MCMC algorithms admit a central limit theorem
regardless of its starting point, i.e., one has, for a Markov chain $(Y_p)$
with invariant distribution $\pi$, and 
and a fixed point $y_1$,
\begin{equation*}
  \left. \sqrt P \pr{\frac 1 P \sum_{p=1}^P \varphi(Y_p) - \pi(\varphi)} \right
  \vert Y_1 = y_1 \Rightarrow \N(0, \sigma^2)
\end{equation*}
for some $\sigma^2$, as $P \to \infty$. For uniformly ergodic Markov chains,
stronger results hold. For example, for any deterministic sequence
$(y_p)_{p=1}^\infty$:
\begin{equation*}
  \left. \sqrt P \pr{\frac 1 P \sum_{p=1}^P \varphi(Y_p) - \pi(\varphi)}
  \right \vert Y_1 = y_P \Rightarrow \N(0, \sigma^2).
\end{equation*}

If instead of having a single Markov chain, we have $M= M(P)$
chains $(Y_{p}^{m})$, $m=1,\ldots,M$, running in parallel, then, provided
that the number of chains $M$ is negligible compared to their length $P$, it is
possible to average the result of $M$ chains to get a better one. Specifically,
it can be shown that for any deterministic sequence $(y_P^m)$ indexed by $m$ and
$p$,
\begin{equation}
    \label{eq:intro_cond_convg}
    \left. \sqrt{MP} \pr{\frac 1 M \sum_{m=1}^M \frac 1 P \sum_{p=1}^P \varphi(Y_p^m)
        - \pi(\varphi)} \right \vert Y_1^{1:M} = y_{P}^{1:M}
    \Rightarrow \mathcal{N}(0, \sigma^2)
\end{equation}
as $P \to \infty$. It is natural to reformulate \eqref{eq:intro_cond_convg}
using the following simplified notation:
\begin{equation}
\label{eq:intro_cond_convg_2}
\left. \sqrt{MP} \pr{\frac 1 M \sum_{m=1}^M \frac 1 P \sum_{p=1}^P \varphi(Y_p^m)
    - \pi(\varphi)} \right \vert Y_1^{1:M} \Rightarrow \N(0, \sigma^2)
\end{equation}
while keeping in mind that $M = M(P)$ and in particular the $\sigma$-algebra
generated by $Y_{1}^{1:M}$ does not stay the same when $P \to \infty$. While the
interpretation \eqref{eq:intro_cond_convg} of the notation of
\eqref{eq:intro_cond_convg_2} is intuitive, a more rigorous formalization will
make manipulations easier. That is the point of the following definition and
lemma, which are simple specific cases of more general results in
\citet{sweeting1989conditional}. The difference with
\citet{sweeting1989conditional} is that we prefer, if possible, to work with
probability conditioned on an event, which is simpler than probability
conditioned on a filtration or a variable.

\begin{defi}[Convergence of conditional distributions]
\label{def:convg_cond_dist}
Let $(X_n)_{n=1}^\infty$ be a sequence of random variables and let
$(\mathcal{F}_n)_{n=1}^\infty$ be a sequence of $\sigma$-algebras (which are not
necessarily nested as in a filtration). We say that the sequence $X_n |
\mathcal{F}_n$ of conditional distributions converge as $n\rightarrow \infty$ to
distribution $\pi$, 
\begin{equation*}
    X_n | \mathcal{F}_n \Rightarrow \pi,
\end{equation*}
if for any sequence $(B_n)_{n=1}^\infty$ of events such that $B_n \in 
\mathcal{F}_n$ and $\P(B_n) > 0$, we have $X_n | B_n \Rightarrow \pi$.
\end{defi}

\begin{lem}
\label{lem:as_cond}
Under the notations of definition~\ref{def:convg_cond_dist}, we have, for any
continuous bounded function $\varphi:\setX\rightarrow \R$,
\begin{equation*}
    \CE{\varphi(X_n)}{\mathcal{F}_n} \cvas \pi(\varphi).
\end{equation*}
\end{lem}

\textbf{Remark.} This result is in fact used in \cite{sweeting1989conditional}
as the \textit{definition} of convergence in conditional probability. As 
said above, we prefer Definition~\ref{def:convg_cond_dist} as we find it
more convenient to work with probability conditioned on events than
probability conditioned on sigma-algebras. 

\begin{proof}
For some $\epsilon>0$, define the events $B_n$ as 
\begin{equation*}
    B_n \eqdef \px{\CE{\varphi(X_n)}{\mathcal{F}_n} - \pi(\varphi) \geq \epsilon}.
\end{equation*}
If $\P(B_n) > 0$, one can write, since $B_n \in \mathcal{F}_n$:
\begin{equation}
\label{eq:as_cond}
    \CE{\varphi(X_n)}{B_n} = \CE{\CE{\varphi(X_n)}{\mathcal{F}_n}}{B_n} \geq \pi(\varphi) + \epsilon.
\end{equation}
If there exists an infinity of $n$ such that $\P(B_n) > 0$, we have by
Definition~\ref{def:convg_cond_dist} that $X_n | B_n \Rightarrow \pi$, which
leads to a contradiction if we let $n \to \infty$ in both sides of
\eqref{eq:as_cond}. Thus, there exists some $n_1$ such that $\P(B_n) = 0
,\forall n \geq n_1 $. Similarly, one may show that there exists $n_2$ such that
$\P(C_n) = 0$, $\forall n \geq n_2$, where $$C_n =
\px{\E[\varphi(X_n)|\mathcal{F}_n] - \pi(\varphi) < -\epsilon}.$$ Now, note that
the desired almost-sure convergence is equivalent to the fact that the random
variable
$$R := \limsup_{n\rightarrow\infty} \abs{\E[\varphi(X_n) | \mathcal{F}_n] - \pi(\varphi)}$$
equals $0$ almost surely. Indeed, for any $\epsilon > 0$, the event $\px{R \geq
  \epsilon}$ is contained in $\pr{\bigcup_{n=n_1}^\infty B_n}$ $\cup$
$\pr{\bigcup_{n=n_2}^\infty C_n}$, which has probability zero.
\end{proof}

\begin{lem}
\label{lem:conv_prod}
Let $(X_n)_{n=1}^\infty$ and $(Y_n)_{n=1}^\infty$ be two sequences of random
variables such that $X_n \Rightarrow \P_X$ and $Y_n | X_n \Rightarrow \P_Y$
where the latter is understood in terms of Definition~\ref{def:convg_cond_dist}.
Then $(X_n, Y_n)$ $\Rightarrow \P_X \otimes \P_Y$.
\end{lem}
\begin{proof}
  Let $Y$ be a $\P_Y$-distributed random variable. We have that
\begin{equation*}
  \abs{\E[e^{iuX_n + ivY_n}] - \E[e^{iuX_n}]\E[e^{ivY}]} = \abs{\E\ps{e^{iuX_n} \pr{\E[e^{ivY_n}|X_n] - \E[e^{ivY}]}}}
\end{equation*}
tends to $0$ by dominated convergence theorem and the fact that
$\E[e^{ivY_n}|X_n] - \E[e^{ivY}]$ converges almost surely to $0$ (Lemma~\ref{lem:as_cond}).
\end{proof}

We are now able to prove Theorem~\ref{thm:clt_norm_cst}. 

\begin{proof}
  The idea of the proof is to show something very similar to
  \eqref{eq:intro_cond_convg}. Indeed, we shall show the following conditional
  version of \eqref{eq:clt_in_P}:
\begin{equation}
\label{eq:cond_convg_at_t}
\left. \frac{1}{\sqrt P} \sum_{p=1}^P \varphi_{M}(Z_p) \right \vert \mathcal{F}_{t-1}
\Rightarrow \N \left( 0, \vas(M_t,\varphi) \right)
\end{equation}
which by Definition~\ref{def:convg_cond_dist} means
\begin{equation}
\label{eq:cond_convg_at_t_bis}
\left. \frac{1}{\sqrt P} \sum_{p=1}^P \varphi_{M}(Z_p) \right \vert B_{t-1}
\Rightarrow \left( 0, \vas(M_t,\varphi) \right)
\end{equation}
for any sequence $B_{t-1}$ (implicitly indexed by $P$) of events such that $B_{t-1}^P \in \mathcal{F}_{t-1}^P$. The
left hand side of \eqref{eq:cond_convg_at_t_bis} can be decomposed into four
terms as in \eqref{eq:decomp_clt}, where now $(Z_p^{\star})$ is a
stationary Markov chain constructed via a maximal coupling of $\Q_{t-1}^{\otimes
  M}$ and the \textit{conditional} (instead of the full) distribution of
$Z_R$. The first, third and the fourth terms of \eqref{eq:decomp_clt}
can be treated exactly as before. The second term tends to $0$ in probability
when $R = N^\epsilon$ for small enough $\epsilon$, because $M = O(N^\alpha)$ for
$\alpha < 1/2$. Thus \eqref{eq:cond_convg_at_t} holds. Applying it for
$\varphi = G_t$ and using the delta method give the convergence of $\sqrt N
(\log\hat \ell_t - \log \ell_t) | \mathcal{F}_{t-1}$ with asymptotic variance
$\vas(M_t, \bar{G}_t)$. Furthermore, note that by Definition~\ref{def:convg_cond_dist}, 
the convergence of $X_N | \mathcal{F}_N$ implies
the convergence of $X_N | \mathcal{F}'_N$ if $\mathcal{F}'_n \subset
\mathcal{F}_n$ for all $n$. Hence
\begin{equation}
\label{eq:log_before}
\left. \sqrt N \pr{\log \ell_t^N - \log \ell_t} \right
\vert \sqrt N \pr{\log  L_{t-1}^N - \log L_{t-1}}
\Rightarrow \N\left( 0, \vas(M_t, \bar{G}_t) \right). 
\end{equation}

We can now proceed by induction. Suppose that the assertion is verified up to
time $t-1$, that is,
\begin{equation}
\label{eq:log_at_t}
\sqrt N \pr{\log L_{t-1}^N - \log L_{t-1}}
\Rightarrow \N\pr{0, \sum_{s=0}^{t-1} \vas(M_s, \bar{G}_s)}. 
\end{equation}
Then, \eqref{eq:log_before}, \eqref{eq:log_at_t} and Lemma~\ref{lem:conv_prod}
prove the assertion at time $t$.
\end{proof}

\subsection{Proof of Proposition~\ref{prop:comparison}}

We first calculate $\wstd{t}{k}(\varphi)$ by using e.g. formula (11.14) in
\cite{SMCbook}: 
\begin{equation} \label{eq:vtphi_smcbook} 
    \wstd{t}{k}(\varphi) = 
    \sum_{s=0}^t \Q_{s-1}\ps{\px{\bar G_s R_{s+1:t} C_t \varphi}^2}
\end{equation} 
where $\bar{G}_t = G_t / r$, 
$R_t(\varphi) \eqdef M_t \bar G_t \varphi$, $R_{s+1:t}
\eqdef R_{s+1} \circ \ldots \circ R_t$, and $C_t(\varphi) \eqdef \varphi -
\Q_t(\varphi)$. Note that $M_t$, $\bar G_t$ and $C_t$ are all linear
functionals.  From the definition of $M_t$, we have 
\[M_t(x_{t-1}, B) 
    = (1- \tilde p_k) \ind_B(x_{t-1}) + \tilde p_k \pi_{t-1}(B), 
\]
with $\tilde{p}_k = 1 - (1-p)^k$, which leads to 
\[\bar G_s R_{s+1:t} C_t \varphi = \bar G_s
\ps{\tilde p_k \Q_s^\star \bar G_{s+1} + (1-\tilde p_k) \bar G_{s+1}} \ldots
\ps{\tilde p_k \Q_{t-1}^\star \bar G_t + (1-\tilde p_k) \bar G_t} C_t \varphi.
\] 
It is easy to fully extend the above expression if one remarks that for any
$l<t$, $\tilde p_k \Q_l^\star \bar G_{l+1:t} C_t \varphi = 0$. Therefore only
terms without any $\tilde p_k \Q_l^\star \bar G_{l+1}$ actually contribute to
the result. Thus 
\[\bar G_s R_{s+1:t} C_t \varphi = (1-\tilde p_k)^{t-s} \bar
G_{s:t} C_t \varphi.
\] 
We can now plug this into \eqref{eq:vtphi_smcbook} and get 
\begin{align*} \wstd{t}{k}(\varphi) &= \sum_{s=0}^t (1-\tilde p_k)^{2(t-s)}
    \Q_{t-1} \ps{\bar G_{s:t}^2 (C_t \varphi)^2} \\ &= \sum_{s=0}^t (1-\tilde
    p_k)^{2(t-s)} \Q_{s-1} \ps{\bar G_{s:t} \frac{1}{r^{t-s+1}} (C_t\varphi)^2}
    \\ &= \sum_{s=0}^t \frac 1 r \ps{\frac{(1-\tilde p_k)^2}{r}}^{t-s}
\Q_t\ps{(C_t \varphi)^2} \\ &= \frac 1 r \sum_{s=0}^t
\pr{\frac{(1-p)^{2k}}{r}}^s \Var_{\Q_t}(\varphi).  
\end{align*} 

We thus see that the variance of the standard SMC sampler evolves
proportionally to the sum of a geometric series and its stability depends on
whether the base of the series is smaller than or greater than $1$. This proves
the second point of the proposition. For the third point, note that
\begin{align*} \mathcal{\tilde V}_t(\varphi) &= \Q_{t-1} \ps{(C_{t-1}\varphi)^2
    + 2\sum_{s=1}^\infty (C_{t-1}\varphi) (K_t^s C_{t-1} \varphi)} \\ &=
    \Q_{t-1} \ps{(C_{t-1}\varphi)^2 + 2 \sum_{s=1}^\infty (C_{t-1}\varphi)^2
(1-p)^s} \\ &= \pr{\frac 2 p -1} \Q_{t-1}\ps{(C_{t-1}\varphi)^2}, 
\end{align*}
from which 
\begin{align*} 
    \wwf{t}(\varphi) &= \mathcal{\tilde V}_t(\bar G_t C_t \varphi) \\ 
    &= \pr{\frac 2 p -1} \Q_{t-1}\ps{(C_{t-1}\bar G_t C_t \varphi)^2} \\ 
    &= \pr{\frac 2p-1} \Q_{t-1} \ps{(\bar G_t C_t \varphi)^2} \\
    &= \frac 1r \pr{\frac 2 p -1} \Var_{\Q_t}(\varphi).  
\end{align*} 

Finally, to prove the last point of the proposition, we write 
\begin{equation} \label{eq:time_lim_bound} 
    \lim_{t\to\infty}
    \frac{\ifactorwf{t}}{k\ifactorstd{t}{k}} = \frac{r^{-1} (\frac 2 p
    -1)}{r^{-1}k \pr{1-\frac{(1-p)^{2k}}{r}}^{-1}} \leq \pr{\frac 2 p -1}
    \frac{1-(1-p)^{2k}}{k} 
\end{equation} 
as the second to last expression is non-decreasing in $r$. Next, consider the
function $f(p) \eqdef (1-p)^{2k} + 2kp$ of which the derivative $f'(p) =
2k(1-(1-p)^{2k-1})$ is non-negative thanks to the fact that $k \geq 1$.  We
have $f(p) \geq f(0)=1$, which, when plugged into Equation
\eqref{eq:time_lim_bound}, gives 
\[ \lim_{t\to\infty} \frac{\ifactorwf{t}}{k\ifactorstd{t}{k}} \leq \pr{\frac 2
p -1} \frac{2kp}{k} \leq 4.  \]

\bibliography{bib/complete.bib} \bibliographystyle{apalike}

\end{document}
